%% file: main.tex
\documentclass[11pt]{article}

\usepackage{amsthm}
\usepackage{graphicx} 
\usepackage{array} 

\usepackage{amsmath, amssymb, amsfonts, verbatim}
\usepackage{hyphenat,epsfig,subfigure,multirow}

\usepackage[usenames,dvipsnames]{xcolor}
\usepackage[ruled]{algorithm2e}

\usepackage{tcolorbox}
\tcbuselibrary{skins,breakable}
\tcbset{enhanced jigsaw}

\usepackage[normalem]{ulem}
\usepackage[compact]{titlesec}

\definecolor{DarkRed}{rgb}{0.5,0.1,0.1}
\definecolor{DarkBlue}{rgb}{0.1,0.1,0.5}

\usepackage{nameref}
\definecolor{ForestGreen}{rgb}{0.1333,0.5451,0.1333}
\definecolor{Red}{rgb}{0.9,0,0}
\usepackage[linktocpage=true,
	pagebackref=true,colorlinks,
	linkcolor=DarkRed,citecolor=ForestGreen,
	bookmarks,bookmarksopen,bookmarksnumbered]
	{hyperref}
\usepackage[noabbrev,nameinlink]{cleveref}
\crefname{property}{property}{Property}
\creflabelformat{property}{(#1)#2#3}
\crefname{equation}{eq}{Eq}
\creflabelformat{equation}{(#1)#2#3}

\usepackage{bm}
\usepackage{url}
\usepackage{xspace}
\usepackage[mathscr]{euscript}

\usepackage{tikz}
\usetikzlibrary{arrows}
\usetikzlibrary{arrows.meta}
\usetikzlibrary{shapes}
\usetikzlibrary{backgrounds}
\usetikzlibrary{positioning}
\usetikzlibrary{decorations.markings}
\usetikzlibrary{patterns}
\usetikzlibrary{calc}
\usetikzlibrary{fit}
\usetikzlibrary{snakes}

\usepackage{mdframed}

\usepackage[noend]{algpseudocode}
\makeatletter
\def\BState{\State\hskip-\ALG@thistlm}
\makeatother

\usepackage{cite}
\usepackage{enumitem}

\usepackage[margin=1in]{geometry}

\newtheorem{theorem}{Theorem}
\newtheorem{lemma}{Lemma}[section]
\newtheorem{proposition}[lemma]{Proposition}

\newtheorem{claim}[lemma]{Claim}

\newtheorem{definition}{Definition}

\newtheorem*{claim*}{Claim}
\newtheorem*{proposition*}{Proposition}
\newtheorem*{lemma*}{Lemma}
\newtheorem*{problem*}{Problem}

\newtheorem{mdresult}{Result}
\newenvironment{result}{\begin{mdframed}[backgroundcolor=lightgray!40,topline=false,rightline=false,leftline=false,bottomline=false,innertopmargin=2pt]\begin{mdresult}}{\end{mdresult}\end{mdframed}}

\theoremstyle{definition}
\newtheorem{remark}[lemma]{Remark}

\allowdisplaybreaks

\renewcommand{\qed}{\nobreak \ifvmode \relax \else
      \ifdim\lastskip<1.5em \hskip-\lastskip
      \hskip1.5em plus0em minus0.5em \fi \nobreak
      \vrule height0.75em width0.5em depth0.25em\fi}

\newcommand*\samethanks[1][\value{footnote}]{\footnotemark[#1]}

\newcommand{\logstar}[1]{\ensuremath{\log^{*}\!{#1}}}

\input{macros}

\title{Exploration with Limited Memory: Streaming Algorithms for Coin Tossing, Noisy Comparisons, and Multi-Armed Bandits}
\author{Sepehr Assadi\footnote{Department of Computer Science, Rutgers University. Email: \texttt{\{sepehr.assadi,chen.wang.cs\}@rutgers.edu}} \and Chen Wang\samethanks}

\date{}

\begin{document}
\maketitle

\pagenumbering{roman}

\input{abstract}

\clearpage

\setcounter{tocdepth}{3}
\tableofcontents

\clearpage

\pagenumbering{arabic}
\setcounter{page}{1}

\input{intro}
\input{problem}

\input{prelim}
\input{main-alg}
\input{top-k}
\input{noisy-comp}

\input{eps-best}
\subsection*{Acknowledgement} 
We would like to thank Arpit Agarwal, Sanjeev Khanna, Shay Moran, and David Woodruff for helpful discussions. We are also grateful to the anonymous reviewers of STOC 2020 for many helpful comments that helped with the presentation of this paper (including the addition of~\Cref{sec:rand-walk}).

\bibliographystyle{abbrv}
\bibliography{general}

\clearpage
\appendix

\part*{Appendix}

\input{prelim-algs}
\input{random-walk}

\end{document}

%% file: macros.tex
\newcommand{\eps}{\ensuremath{\varepsilon}}
\newcommand{\Paren}[1]{\Big(#1\Big)}

\newcommand{\bracket}[1]{\left[#1\right]}
\newcommand{\paren}[1]{\ensuremath{\left(#1\right)}\xspace}
\newcommand{\card}[1]{\left\vert{#1}\right\vert}

\newcommand{\ceil}[1]{{\left\lceil{#1}\right\rceil}}

\newcommand{\expect}[1]{\Exp\bracket{#1}}

\newcommand{\set}[1]{\ensuremath{\left\{ #1 \right\}}}
\newcommand{\poly}{\mbox{\rm poly}}

\DeclareMathOperator*{\Exp}{\ensuremath{{\mathbb{E}}}}
\DeclareMathOperator*{\Prob}{\ensuremath{\textnormal{Pr}}}
\renewcommand{\Pr}{\Prob}
\newcommand{\EX}{\Exp}

\newenvironment{tbox}{\begin{tcolorbox}[
		enlarge top by=5pt,
		enlarge bottom by=5pt,
		 breakable,
		 boxsep=0pt,
                  left=4pt,
                  right=4pt,
                  top=10pt,
                  arc=0pt,
                  boxrule=1pt,toprule=1pt,
                  colback=white
                  ]
	}
{\end{tcolorbox}}

\newcommand{\event}{\ensuremath{\mathcal{E}}}

\newcommand{\PAC}{\ensuremath{\textsf{PAC}}}

%% file: abstract.tex

\begin{abstract}

Consider the following abstract coin tossing problem: Given a set of $n$ coins with unknown biases, find the most biased coin using a minimal number of coin tosses.  
This is a common abstraction of various exploration problems in theoretical computer science and machine learning  and has been studied extensively over the years. In particular, algorithms with optimal \emph{sample complexity} (number of coin tosses) have been known for this problem for quite some time.

\medskip

Motivated by applications to processing massive datasets, we study the \emph{space complexity} of solving this problem with optimal number of coin tosses in the \emph{streaming} model. 
In this model, the coins are arriving one by one and the algorithm is only allowed to store a limited number of  coins at any point -- any coin not present in the memory is lost and can no longer be tossed or compared to arriving coins. 
Prior algorithms for the coin tossing problem with optimal sample complexity are  based on iterative elimination of coins which inherently require storing all the coins,  leading to memory-inefficient streaming algorithms. 

\medskip

We remedy this state-of-affairs by presenting a series of improved streaming algorithms for this problem: we start with a simple algorithm which require storing only $O(\log{n})$ coins and then iteratively 
refine it further and further, leading to algorithms with $O(\log\log{(n)})$ memory, $O(\logstar{(n)})$ memory, and finally a one \emph{that only stores a single extra coin in memory} -- 
the same exact space needed to just store the best coin throughout the stream.

\medskip

Furthermore, we extend our algorithms to the problem of finding the $k$ most biased coins as well as other exploration problems such as finding top-$k$ elements using noisy comparisons or finding an $\eps$-best arm in stochastic multi-armed bandits, and 
obtain  efficient streaming algorithms for these problems. 

\end{abstract}

%% file: intro.tex

\newcommand{\coin}{\ensuremath{\textnormal{\textsf{coin}}}\xspace}

\newcommand{\arm}{\ensuremath{\textnormal{\textsf{arm}}}\xspace}

\newcommand{\armk}[1]{\ensuremath{\arm_{[#1]}}\xspace}
\newcommand{\muk}[1]{\ensuremath{\mu_{[#1]}}\xspace}
\section{Introduction}\label{sec:intro}

Suppose you are given $n$ coins with unknown biases; how many \emph{samples} (coin tosses) are needed to find the most biased coin with a large (constant) probability of success? This basic problem
captures the essence of various \emph{(pure) exploration} problems in theoretical computer science and machine learning in which the general goal is to find a best option among a set of alternatives using a minimal number of stochastic/noisy trials. 
Examples include rank aggregation with noisy comparisons (e.g.~\cite{FeigeRPU94,BusaFeketeSCWH13,DavidsonKMR14,BravermanMW16,ShahW17,ChenGMS17,ShahW17,ChenLM18,CohenAddadMM20}),  best 
arm identification in multi-armed bandits (e.g.~\cite{EvenDarMM02,MannorT03,AudibertBM10,KalyanakrishnanS10,KarninKS13,JamiesonMNB14,KaufmannCG16,CarpentierL16,ChenLQ17}), or computing with noisy decision trees (e.g.~\cite{ReischukS91,FeigeRPU94,Newman04,GoyalS10}). These problems in turn have a wide range 
of applications in medical trials~\cite{Robbins1952}, networking~\cite{Xue2016,TalebiZCPJ18},  web search~\cite{DworkKNS01}, crowdsourcing~\cite{ChenBCH13,ZhouCL14}, and display advertising~\cite{AgarwalCEMPRRZ08}, among  others. 

This coin tossing problem admits a  natural solution: sample/toss each coin ``enough'' number of times so that the empirical bias of each coin ``closely'' matches its true bias; then find the coin with the most empirical bias. Assuming there 
is some \emph{constant known gap}  between the bias of the most and the second most biased coins, a simple argument suggests that tossing each coin $O(\log{n})$ times is enough for this purpose, leading to an algorithm with $O(n\log{n})$ coin tosses overall. 

It turns out that one can beat this natural approach and solve the problem with $O(n)$ samples~\cite{EvenDarMM02} (see also~\cite{FeigeRPU94}) which is the (asymptotically) optimal \emph{sample complexity} of this 
problem~\cite{MannorT03}. Sample-optimal algorithms for this problem has since been studied extensively in various directions: finding multiple coins (e.g.~\cite{KalyanakrishnanS10,KalyanakrishnanTAS12}), with combinatorial
constraints (e.g.~\cite{ChenLKLC14,ChenGL16}), instance-optimal algorithms (e.g.~\cite{JamiesonMNB14,CohenAddadMM20}), fixed-budget algorithms (e.g.~\cite{BubeckMS11,BubeckWV13,CarpentierL16}), limited adaptivity algorithms 
(e.g.~\cite{GoyalS10,AgarwalAAK17,CohenAddadMM20}), or collaborative learning algorithms (e.g.~\cite{HillelKKLS13,TaoZZ19,BlumHPQ17}), to mention a few.

Alas, the sample-efficiency of these algorithms comes at a certain cost: unlike the 
basic approach that processes the coins ``on the fly'' by storing the current candidate coin, these more complicated algorithms need to store all coins and revisit them frequently before making a decision. 
As such, these solutions can be prohibitively expensive in their memory requirement in applications with a massive number of coins/options (including several of above examples). In such
scenarios, the \emph{space complexity}, in addition to the sample complexity, plays a major role in  the efficiency of algorithms. 

The \emph{streaming model} of computation, pioneered by~\cite{AlonMS96, HenzingerRR98, FeigenbaumKSV99}, precisely captures these scenarios. 
In this model, the coins are arriving one by one and the algorithm is only allowed to store a limited number of coins at any point -- any coin not present in the memory is lost and can no longer be tossed or compared to arriving coins. 
We refer to the maximum number of coins stored by the algorithm at any point during the stream as the space complexity or memory cost of the algorithm (see \Cref{sec:prelim} for details). 
We can now ask the following fundamental question: 
\vspace{-0.1cm}
\begin{quote}
\emph{What is the memory cost of achieving (asymptotically) optimal sample complexity for the coin tossing problem in the streaming model?}
\end{quote}
\vspace{-0.2cm}

Our main (conceptual) finding in this paper is that, surprisingly, there is almost \emph{no tradeoff} between sample-efficiency and space-efficiency for coin tossing: one can achieve the sharpest possible bound on the space complexity, namely a memory of \emph{a single extra coin}, without having to settle for an asymptotically sub-optimal sample complexity!

We further build on this result to design streaming algorithms for finding multiple coins with largest biases
and for other related problems 
such as partitioning totally ordered elements using noisy comparisons or finding approximate best arms in stochastic multi-armed bandits. The extension of our coin tossing results to noisy comparisons  is 
particularly interesting as there is no black-box reduction between the two models and indeed these models are often considered conceptually related but disjoint technique-wise (see, e.g.~\cite{BravermanMW16,ChenLM18,CohenAddadMM20}). 

\subsection{Our Contributions} 

\paragraph{Most Biased Coin.} Our first main result is a complete resolution of the aforementioned question: 

\begin{result}\label{res:main1}
	There exists a streaming algorithm that achieves the (asymptotically) optimal sample complexity for the coin tossing problem by storing only \textbf{a single extra coin} in its memory. 
\end{result}

We formalize \Cref{res:main1} in \Cref{thm:main}. We emphasize that in~\Cref{res:main1} and throughout the paper, we assume the algorithm knows the gap between the bias of the most  and the second most biased coins. In our recent follow-up work \cite{AW22Neurips}, it is shown that the knowledge of the gap is \emph{necessary}, as otherwise there are instances forcing an unbounded sample complexity.

An interesting byproduct of using just a single-coin memory in~\Cref{res:main1} is that the 
algorithm necessarily maintains the most biased coin as its only candidate once this coin is observed in the stream, namely, it is also an \emph{online} algorithm (this corresponds to the notion of \emph{streaming online algorithms} proposed in~\cite{Sublinear73}.) 

En route to proving \Cref{res:main1}, we design a series of streaming algorithms with optimal sample complexity for coin tossing (see~\Cref{sec:prelim-algs}). We start with a simple algorithm that uses $O(\log{n})$ memory 
by giving a streaming friendly implementation of the {median-elimination} algorithm of~\cite{EvenDarMM02} using the ``merge-and-reduce'' technique from the streaming literature (see, e.g.~\cite{GuhaMMO00,AgarwalHV04}). 
We then show that one can further improve the memory down to $O(\log\log{n})$ coins by designing a variant of merge-and-reduce tailored directly to the coin tossing problem. This adaptation in turn allows us 
to use the more recent aggressive-elimination algorithm of~\cite{AgarwalAAK17} in place of the original median-elimination and reduce the space down to $O(\logstar{(n)})$ coins\footnote{None of these algorithms follow as a black-box 
from prior work and several new ingredients are still needed to make these parts work in the streaming model which can be of their own independent interest. Considering this, and to provide further insight into our main algorithm, we present these intermediate algorithms also in \Cref{sec:prelim-algs}.}. 
The final leap from $O(\logstar{(n)})$ memory algorithm to our  single-coin memory algorithm however is the key step as explained below. 

The memory bound of our intermediate streaming algorithms is heavily tailored to the number of elimination rounds of base algorithms in~\cite{EvenDarMM02,AgarwalAAK17} and it is known that $\Theta(\logstar{(n)})$ bound on number
of elimination rounds is \emph{tight}~\cite{GoyalS10,AgarwalAAK17}. As such, to obtain our final algorithm, we almost entirely forego the elimination approach and devise a new \emph{budgeting strategy} for the problem: we maintain a candidate coin, called the ``king'', throughout the stream and assign it a certain budget which is increased per each new arriving coin and decreased whenever we toss any coin. 
Each arriving coin then ``challenges'' the king by tossing both the king and arriving coin, according to a carefully chosen rule, until either king wins against the new coin (by having a higher empirical bias at any of these challenges) or the budget of the
 king is depleted 
in which case we replace the king with the new coin and restart the process with this new king on the remainder of the stream. 

This budgeting allows us to use a basic amortized analysis and argue that the total number of coin tosses by the algorithm is still $O(n)$ (albeit with a much more chaotic pattern of samples per coin compared to elimination-based algorithms). The key challenge is however to ensure that once the most biased coin becomes the king, it will not exhaust its budget throughout the remaining 
length of the stream which can be $\Theta(n)$-long. This requires proving that the random variable corresponding to the remaining budget of the king does not have any significant deviation from its expectation \emph{throughout} the entire
length of the stream and not only at \emph{any fixed} point. This is similar-in-spirit to the fact that a length $n$ symmetric $(\pm 1)$-random walk on a line does not deviate from the $\Theta(\sqrt{n})$ bound implied by the variance not only at the end, 
but throughout the entire walk. The proofs are however different  since our version of ``random walk'' includes unbounded step sizes. As such, we first prove that these step sizes form a sub-exponential distribution, and then
use Bernstein's inequality to prove the desired concentration bound. The more `flexible' version of random walk can be of independent interests, and we provide a general characterization for sub-exponential step size random walk in \Cref{sec:rand-walk}.

\paragraph{Top-$k$ Most Biased Coins.} A standard generalization of the coin problem we discussed so far is to find the top-$k$ most biased coins assuming a gap between the bias of the $k$-th and $(k+1)$-th most biased
coin.  This problem has also been studied extensively in the literature and it is known that the (asymptotically) optimal sample complexity for this problem is $\Theta(n\log{k})$~\cite{KalyanakrishnanS10,KalyanakrishnanTAS12}. 
We show that this optimal sample complexity
 can be achieved by memory-efficient streaming algorithms. 

\begin{result}\label{res:main2}
	There exists a streaming algorithm that achieves the (asymptotically) optimal sample complexity for finding the top-$k$ most biased coins by storing only $O(k)$ coins in the memory.  
\end{result}
We formalize \Cref{res:main2} in \Cref{thm:top-k}. It is clear that any streaming algorithm for this problem requires memory of $k$ coins to simply store the answer. As such, \Cref{res:main2} implies that one can simultaneously achieve
the \emph{asymptotic} optimal memory and sample complexity for this problem. 

 The starting point of this algorithm is our budgeting approach in~\Cref{res:main1}. However, 
there are two main challenges that need to be addressed:  (1) we now need to maintain $k$  ``kings'' but can no longer compare each arriving coin with (or assign
a unit of budget to) every king (otherwise, there will be $\Omega(nk)$ coin tosses); more importantly (2) we need to collect all the top-$k$ coins and still cannot guarantee any suitable (probabilistic) outcome while comparing any of 
these two coins to each other (as there may not be any gap between their biases in general). We elaborate on these challenges and how we address them in the high level overview of our algorithm in \Cref{sec:top-k} and only mention here that 
addressing these challenges turn out to be a highly non-trivial task and in fact our algorithm in~\Cref{res:main2} is the main \emph{technical} contribution of our work. 

\subsection*{Application to Noisy Comparison Model} 

An interesting application of our results is to the following noisy comparison problem: we have a collection of $n$ elements with an \emph{unknown total order} and we can compare any two element $i$ and $j$ according to a \emph{noisy} version of
this ordering: when comparing $i,j$, with probability $2/3$ we receive the true answer whether $i < j$ or $j < i$, and with the remaining probability, the answer is arbitrarily. The goal is to partition the input into the set of $k$ largest element
and $(n-k)$ remaining smaller elements. This problem, often referred to as the \emph{partition} problem, has received a burst of interest in recent 
years (see, e.g.~\cite{BravermanMW16,ChenGMS17,ChenLM18,CohenAddadMM20} and references therein). The streaming version of this problem, when the elements are arriving one by one in the stream and 
only the elements stored in the memory can be compared, is equally well-motivated (see~\cite{BravermanMW16} for related applications). 

It is easy to spot a fundamental difference between the partition problem and coin tossing: the first one uses \emph{ordinal} information between the elements while the latter concerns \emph{cardinal} information. 
Due to this difference, the algorithms in one model do not carry over to another and the research on these two problems has been mostly disjoint (see, e.g.~\cite{BravermanMW16,CohenAddadMM20} -- see also~\cite{AgarwalAAK17} that gives a black-box 
reduction from coin tossing to a \emph{different} noisy model of comparison and~\cite{CohenAddadMM20} that shows this, or any other, reduction cannot work in the model studied in our paper). 

Interestingly, our algorithms in~\Cref{res:main1} and~\Cref{res:main2} operate by only comparing empirical biases of coins directly with each other (through the notion of ``challenging'' described above), which is an ordinal information. Rather more formally, 
our algorithms work even if instead of sampling the coins and observing their empirical biases, they can sample two coins and observe which one has the higher empirical bias. Owing to this property, we can indeed extend our algorithms in these
results to the partition problem in the noisy model and obtain the following result.  
 
\begin{mdresult}\label{res:noisy}
	There exists a streaming algorithm for the partition problem that uses $O(n\log{k})$ noisy comparisons and a memory of $O(k)$ elements (the memory is a single extra element when $k=1$). 
\end{mdresult}
\Cref{res:noisy} is formalized in \Cref{thm:comp}, presented in~\Cref{sec:comp}. Considering that the (asymptotically) optimal number of samples for the partition problem is $O(n\log{k})$~\cite{CohenAddadMM20}, \Cref{res:noisy} achieves the asymptotically optimal sample complexity and space complexity simultaneously. 

\subsection*{Application to Stochastic Multi-Armed Bandits}

The $\eps$-best arm identification (or PAC-learning) in the stochastic multi-armed bandit (MAB) games is defined as follows: we have a collection of $n$ arms with unknown reward distributions in $[0,1]$; the algorithm can pull (sample) each arm and receive a reward from the 
corresponding distribution. The goal is to, given a parameter $\eps \in (0,1)$, find any arm with expected reward at most $\eps$ less that the expected reward of the best arm, referred to as an \emph{$\eps$-best arm}. 
This problem is a (pure) exploration variant of the more general \emph{regret minimization} problem in MABs introduced more than half a century ago~\cite{Robbins1952} and has been studied extensively on its own (see, e.g.~\cite{EvenDarMM02,MannorT03,AudibertBM10,KalyanakrishnanS10,KalyanakrishnanTAS12,KarninKS13,JamiesonMNB14,KaufmannCG16,CarpentierL16,ChenLQ17} and references therein). Again, the streaming model 
for this problem, in which the arms are arriving one by one and can only be pulled if they are stored explicitly in the memory, is highly motivated; see, e.g., the recent work of~\cite{LiauSPY18,ChaudhuriK19} on 
a related model to streaming and the classical work of~\cite{Cover68} (we will elaborate on the connection between our work and the first two below). 

It is easy to see that the coin tossing problem is a special case of this problem when the reward distributions are Bernoulli and more importantly, there is a \emph{gap} of $\eps$ between the expected reward of the best arm
and any other arm (making the $\eps$-best arm unique). In general, these differences do not matter much and most algorithms for the coin tossing problem appear to extend directly to the $\eps$-best arm problem as well. 
Unfortunately however, this is \emph{not} the case for  our algorithm in~\Cref{res:main1} (the brief intuition is that our algorithm only considers
ordinal information between the empirical biases and a set of arms with gradually decreasing expected reward can ``fool'' the algorithm -- we discuss this in detail in~\Cref{sec:eps-best}). 
Nevertheless, we observe that we can extend our $O(\logstar{(n)})$ memory algorithm for coin tossing to this problem and prove the following result. 

\begin{mdresult}\label{res:best-arm}
	There exists a streaming algorithm for $\eps$-best arm identification in stochastic multi-armed bandits that uses $O(n/\eps^2)$ arm pulls and a memory of $O(\logstar{(n)})$ arms. 
\end{mdresult}

\Cref{res:best-arm} is formalized in~\Cref{thm:eps-best-logstar}, presented in~\Cref{sec:eps-best}. 
The sample complexity of this algorithm is asymptotically optimal~\cite{MannorT03} but its memory is within a \emph{non-constant} (albeit extremely small\footnote{Recall that for every realistic input size $n$, $\logstar{(n)} \leq 5$.}) factor 
of the (best known) bounds.

\begin{remark}\label{rem:bug}
	The earlier version of this paper had an error in the proof of~\Cref{res:best-arm} that was discovered in the recent work of~\cite{MaitiPK21} (who also gave an excellent example showing that the error is not limited to the analysis and 
	the original algorithm is indeed incorrect); we note that this error is entirely independent of the main results of this paper for the coin tossing problem and is due to a subtle difference between coin tossing and 
	stochastic multi-armed bandit. 
	The authors of~\cite{MaitiPK21} gave a different algorithm and recovered~\Cref{res:best-arm}. In this version, we show  how a simple modification of our previous algorithm 
	can  circumvent this error and obtain a correct proof of~\Cref{res:best-arm}. 
\end{remark}

In the new version, we further show that it is indeed possible to adapt our algorithm in~\Cref{res:main1} to the $\eps$-best arm exploration, albeit with some more involved technical considerations. The new result is shown as follows.

\begin{mdresult}\label{res:best-arm-eps}
	There exists a streaming algorithm for $\eps$-best arm identification in stochastic multi-armed bandits that uses $O(n/\eps^2)$ arm pulls and a memory of $2$ arms as long as $\eps > n^{-0.999}$. 
\end{mdresult}

\Cref{res:best-arm-eps} achieves asymptotically optimal sample complexity and memory complexity (it is slightly worse than the sharp optimal memory complexity of 1 arm) as long as the approximation factor $\eps$ is not extremely small. Together with a recent algorithm of \cite{JinH0X21}, it settles a key open problem from the conference version of this paper as well as~\cite{MaitiPK21}.  

\subsection{Recent Related Work on Streaming Coin Tossing} 

We conclude this section by discussing the connection between our work and the recent progress of the streaming coin tossing problem. Before the conference version of our work, there has been some study on streaming coin tossing (and multi-armed bandits) on \emph{regret minimization} (\!\cite{LiauSPY18,ChaudhuriK19}). 
We note that the regret minmization bounds are not directly comparable to ours, and the algorithms in \cite{LiauSPY18,ChaudhuriK19} are in the multi-pass setting (the algorithm of~\cite{ChaudhuriK19} additionally requires random-order arrival). As such, our work is the first to study the \emph{pure exploration} streaming coin tossing problem. After the publication of our conference paper, there has been a flurry of recent papers studying the pure exploration problem \cite{JinH0X21,MaitiPK21,AW22Neurips,AgarwalKP22}, and the single-pass upper and lower bounds for pure exploration are mostly understood. In particular, by the recent follow-up work of \cite{AW22Neurips}, it is known that in a single pass, the assumption of $\Delta$ is essential for any \emph{bounded} sample complexity, and the worst-case optimal bound (as opposed to the instance-optimal bounds as in \cite{KarninKS13,JamiesonMNB14}) is necessary for many families of instances. On the other hand, if we relax the setting to multi-pass algorithms, \cite{JinH0X21} shows that we can obtain an algorithm with the instance-optimal sample complexity and a single-arm memory in $O(\log(\frac{1}{\Delta}))$ passes, even without the knowledge of $\Delta$. In light of this, the tight upper and lower bounds for multi-pass exploration seem to be an interesting direction to pursue.


%% file: problem.tex

\newcommand{\Mk}[2]{\ensuremath{\mathbb{M}_{#2}(#1)}}

\newcommand{\king}{\ensuremath{\textnormal{\textsf{king}}}\xspace}
\newcommand{\King}{\ensuremath{\king}\xspace}

\newcommand{\mainalg}{\ensuremath{\textnormal{\textsc{Game-Of-Coins}}}\xspace}

\newcommand{\phat}{\widehat{p}}
\newcommand{\rhat}{\widehat{r}}
\newcommand{\muhat}{\widehat{\mu}}

\newcommand{\Budget}{\ensuremath{\Phi}}

\newcommand{\coinstar}{\ensuremath{\textnormal{\textsf{coin}}^{*}}\xspace}

\newcommand{\budget}[1]{\ensuremath{\Budget(#1)}}

\section{Problem Definition: Streaming Coin Tossing}\label{sec:prelim}


In the coin tossing problem that we study, there is a collection of $n$ coins $\set{\coin_i}_{i=1}^{n}$ with unknown biases $\set{p_i}_{i=1}^{n}$ and our goal is to identify the most biased coin, denoted by $\coinstar$, 
via tosses of the coins. We refer to the number of coin tosses by the algorithm as its \emph{sample complexity}. An important parameter that governs the sample complexity of the algorithms is the 
\emph{gap parameter} $\Delta$ which denotes the difference between the bias of the most  and the second most biased coins.  
We assume $\Delta > 0$ and is given to the algorithm -- both assumptions are common in the literature~\cite{EvenDarMM02,KalyanakrishnanS10,ChenS15,ShahW17}. 
 Indeed, the first assumption can be easily lifted by simply re-defining this value to be the gap between bias of the most biased coin and the next \emph{distinct} bias. As for the second assumption, in both applications of our results, this parameter corresponds to the standard \emph{input} parameters of the problem, namely 
 the noise factor $\gamma$ and the approximation factor $\eps$.

We study this problem in the streaming model: The coins are arriving one by one in a stream and the algorithm needs to store each coin explicitly if it wants to 
toss it at some later point in the stream as well. In other words, the algorithm only has access to a coin if this is the current coin arriving in the stream, or the coin is currently stored in the memory of the algorithm. 
Moreover, once a coin is no longer in the memory  (because it was either not stored in the first place or was later replaced by another coin), the algorithm has no further access to this coin (i.e., can neither toss it nor 
bring it back to the memory). We refer to the maximum number of coins stored by the algorithm at any point during the stream as the \emph{space complexity} of the algorithm. 

\begin{remark}
We stated the space complexity of the streaming algorithms in terms of number of stored arms and ignored the other information stored by them. 
This is the standard definition for streaming problems that assume \emph{oracle} access to input (the coin tossing oracle for our purpose) such as streaming algorithms
for submodular optimization (see, e.g.~\cite{BadanidiyuruMKK14,MitrovicBNTC17,KazemiMZLK19}). \emph{All} our algorithms only require to store additional $\Theta(\log{n} + \log{(1/\eps)})$ bits ($O(1)$ words of space in the word-RAM model) per
each coin \emph{in their memory}. We also remark that our $O(\logstar{(n)})$ space algorithm appears to be even implementable with only $\Theta(\log\log{n} + \log{(1/\eps)})$ 
 bit overhead per each memory coin by using the classical noisy counter of~\cite{Morris78a}; however, we do not pursue this direction in this paper. 
\end{remark}

%% file: prelim.tex

\section{Preliminaries}\label{sec:concentration}

We say that a random variable $X$ is \emph{sub-exponential} with parameter $\kappa > 0$, if 
\begin{align}
\Pr\paren{\card{X} \geq t} \leq 2\exp\paren{-t/\kappa} \quad \textnormal{for all $t \geq 0$.} \label{eq:sub-exponential}
\end{align}
The following is a variant of Bernstein's inequality (see~\cite[Proposition 2.7.1 and Theorem 2.8.1]{Vershynin18}). 

\begin{proposition}[Bernstein's inequality; cf.~\cite{Vershynin18}]\label{prop:bernstein}
	Let $X_1,\ldots,X_m$ be $m$ independent, mean zero, sub-exponential random variables with parameter $\kappa > 0$. Then, for every $t > 0$, 
	\begin{align*}
		\Pr\paren{\card{\sum_{i=1}^{m} X_i} \geq t} \leq 2 \cdot \exp\paren{-c \cdot \min\paren{\frac{t^2}{\kappa^2 \cdot m}, \frac{t}{\kappa}}},
	\end{align*}
	for some absolute constant $c > 0$. 
\end{proposition}
 
We also use the following standard variant of Chernoff-Hoeffding bound. 

\begin{proposition}[Chernoff-Hoeffding bound]\label{prop:chernoff}
	Let $X_1,\ldots,X_m$ be $m$ independent random variables with support in $[0,1]$. Define $X := \sum_{i=1}^{m} X_i$. Then, for every $t > 0$, 
	\begin{align*}
		\Pr\paren{\card{X - \expect{X}} > t} \leq 2 \cdot \exp\paren{-\frac{2t^2}{m}}. 
	\end{align*}
\end{proposition}
\noindent
A direct corollary of this bound that we use in our proofs is the following. 
\begin{lemma}\label{lem:coin-comp}
	Let $\coin_1$ and $\coin_2$ be two different coins with biases $p_1$ and $p_2$. Suppose ${p_1 - p_2} \geq \theta$ and we sample each coin $\frac{K}{\theta^2}$ times to obtain 
	empirical biases $\phat_1$ and $\phat_2$. Then, 
	\begin{align*}
		\Pr\paren{\phat_1 \leq \phat_2} \leq 2 \cdot \exp\paren{-\frac{1}{4} \cdot K}. 
	\end{align*}
\end{lemma}
\begin{proof}
	The proof is standard and is only provided for completeness. Two separate applications of \Cref{prop:chernoff} to empirical bias of each coin implies that: 
	\begin{align*}
		\Pr\paren{\phat_1 \leq p_1 - \theta/2} &\leq \exp\paren{-(\theta/2)^2 \cdot (K/\theta^2)} = \exp\paren{-\frac{1}{4} \cdot K}; \\
		\Pr\paren{\phat_2 \geq p_2 + \theta/2} &\leq \exp\paren{-(\theta/2)^2 \cdot (K/\theta^2)} = \exp\paren{-\frac{1}{4} \cdot K}.
	\end{align*}
	A union bound on the events above plus the fact that $p_1 - p_2 \geq \theta$ now finalizes the proof. 
\end{proof}

%% file: main-alg.tex

\section{Most Biased Coin: A Single-Coin Memory Algorithm}\label{sec:main}

We describe our main algorithm for the most biased coin problem in this section. 

\begin{theorem}[Formalization of \Cref{res:main1}]\label{thm:main}
	There exists a streaming algorithm that given $n$ coins arriving in a stream with the gap parameter $\Delta$ and confidence parameter $\delta$, finds the most biased coin with probability at least $1-\delta$ 
	using $O(\frac{n}{\Delta^2} \cdot \log{(1/\delta)})$ coin tosses and a memory of a single coin. 
\end{theorem}

Note that the sample complexity of our algorithm in  \Cref{thm:main} is asymptotically optimal in all three parameters and its space is minimum possible. 
We start with a high level overview of our algorithm, followed by its description, and then its analysis.  We  refer the  reader to~\Cref{sec:prelim-algs} 
that contains our intermediate streaming algorithms with sub-optimal space complexity as a warm-up to this main algorithm. 

\subsection{High Level Overview} 
The high level strategy of our algorithm is quite intuitive: 
The algorithm maintains a \emph{single} coin in its memory, referred to as $\King$. The goal is to ensure that at the end of the stream $\King$ is the most biased coin. 
Once a new coin arrives in the stream, we toss both the $\King$ and the new coin a certain number of times and based on the empirical bias, we may decide to overthrow the $\King$ 
and let the arriving coin become the new \King. The challenge is of course to implement this intuitive strategy without using a large number of coin tosses. 

A key step in ensuring the sample efficiency is a \emph{lazy} challenging rule (as opposed to the fixed rules in elimination-based algorithms; see~\Cref{sec:prelim-algs}) implemented in multiple \emph{levels}: to compare $\King$ and the newly arrived $\coin$, we 
first toss both coins a certain constant number of times; 
if the empirical bias of $\king$ is already larger than that of $\coin$, we consider $\king$ the winner and move on; otherwise, we go to the next level and repeat this process with a larger number of coin tosses, 
and continue the same way -- we only overthrow the $\king$ if it loses to $\coin$ for a ``large'' number of times (we elaborate more on this  below). We choose the number of  samples in each level to ensure that the following two properties: $(1)$ when the best 
coin  arrives in the stream, it has a large  probability of winning
against any  $\king$ at this point (no matter the budget of the $\king$), and $(2)$ when $\king$ becomes the best coin, it has a small probability of losing to \emph{any} coin afterwards. 

The approach above allows us to argue that with large probability, $\king$ is equal to the best coin at the end of the stream. However, it is still not enough to ensure the sample efficiency of the algorithm, because the lazy challenging rule allows for a large number 
of coin tosses per challenge (this is particularly problematic when $\king$ is \emph{not} the most biased coin). We address this using an {amortized} analysis by allocating certain \emph{budget} to the $\king$: each $\king$ starts with some fixed (constant) budget 
and any new coin that arrives in the stream will increase the budget of $\king$ by some fixed (constant) number; the budget is reduced by one whenever we sample the $\king$ and its challenger. This way, we will simply overthrow the $\king$ once it has exhausted 
its \emph{entire} budget accumulated so far. In that case we let the current challenger  become the new $\king$. The budget is then restarted for the new $\king$ and we continue as before. 

Introduction of this budget ensures the sample efficiency of the algorithm (deterministically). However, we now need to  make sure that the most biased coin will not exhaust its budget as the $\king$ and get overthrown. 
The lazy challenging rule we defined can be used to ensure that once the best coin becomes $\king$, any remaining coin in the stream can only challenge the $\king$ \emph{in expectation} with $O(1)$ samples, hence, by the time we visit the $m$-th
next coin, we have used only $O(m)$ coin tosses \emph{in expectation}, which fits the budget for $\king$. 
But the worry is that \emph{during} a $\Theta(n)$-length stream, there will be times that for which this random variable (the budget used) takes values $ \gg O(m)$ (specially consider the \emph{unboundedness} of tosses per each trial which is necessary to ensure 
correctness). It turns out however this cannot happen and we can prove that 
with high (constant) probability, throughout the \emph{entire stream}, the number of times $\king$ is challenged is linear in the number of challengers. In order to do this, we need to ensure that our challenging rule is ``conservative'' enough (the exact opposite of 
our $O(\logstar{(n)})$ space algorithm in~\Cref{sec:prelim-algs}) so that even though coin tosses per each challenge may be unbounded, they still form a sub-exponential distribution and hence we can apply Bernstein's inequality to prove the desired concentration bound.


\subsection{The Algorithm: \mainalg}

We now present our algorithm $\mainalg$. The input to the algorithm is the set of $n$ coins $\set{\coin_i}_{i=1}^{n}$ arriving in an arbitrary order in a stream,  the gap parameter $\Delta > 0$, and the confidence parameter $\delta \in (0,1)$ 
(the algorithm does \emph{not} need to know the value of $n$ in advance). Let us first set up the following parameters:
\begin{align*}
	\qquad \qquad \qquad \qquad \qquad &\!\!\!\set{r_\ell}_{\ell=1}^{\infty}: \quad r_\ell = 3^{\ell};
	\\ \tag{intermediate parameters to define the number of samples per each level of the challenge} \\
	&\!\!\!\set{s_\ell}_{\ell=1}^{\infty}: \quad s_\ell :=  \frac{4}{\Delta^2} \cdot \ln{(1/\delta)} \cdot r_\ell;   \\ \tag{the number of samples per each level of the challenge} \\
	&b := \frac{4}{\Delta^2} \cdot C \cdot \ln{(1/\delta)} + s_1.  \\ \tag{the budget given to the $\king$ per each new coin -- $C > 0$ is a constant to be determined later}
\end{align*}
We are now ready to present the algorithm: 

\begin{tbox}
	\textbf{Algorithm} $\mainalg$: 
	\begin{enumerate}[label=($\arabic*$)]
		\item\label{line1} Let $\king$ be the first available coin and set its budget $\Budget := \budget{\king} = 0$. 
		\item For each arriving $\coin_i$ in the stream do: 
		\begin{enumerate}[label=($\alph*$)] 
			\item Increase the budget $\budget{\king}$ by $b$. 
			\item \textbf{\underline{Challenge subroutine}:} For level $\ell = 1$ to $+\infty$ do: 
			\begin{enumerate}
				\item If $\budget{\king} < s_\ell$: we declare $\king$ defeated and go to Line~\ref{line1}.
				\item Otherwise, we decrease $\budget{\king}$ by $s_\ell$ and sample both $\king$ and $\coin_i$ for $s_\ell$ times. 
				\item Let $\phat_{\king}$ and $\phat_{i}$ denote the empirical biases of $\king$ and $\coin_i$ in this trial. 
				\item If $\phat_{\king} > \phat_i$, we declare $\king$ winner and go to the next coin in the stream; otherwise, we go to the next level of the challenge (increment $\ell$ by one). 
			\end{enumerate}
		\end{enumerate}
		\item Return $\king$ as the best coin in the stream. 
	\end{enumerate}
\end{tbox}
\noindent
This concludes the description of our algorithm. The sample complexity of this algorithm can be bounded easily using an amortized analysis. 
\begin{claim}\label{clm:main-sample}
	The total number of coin tosses by the algorithm is at most $4n \cdot b = O(\frac{n}{\Delta^2} \cdot \log{(1/\delta)})$. 
\end{claim}
\begin{proof}
	The proof is a straightforward amortized analysis. Each arriving coin in the stream can increase the budget by $b$ and each time we make a new $\king$ we allocate another $b$ budget to it so over all
	we increase the budget by at most $2n \cdot b$ in total. On the other hand, each unit of budget is responsible for two coin tosses (for the $\king$ and its challenger) and so 
	the total number of coin tosses is at most $4n \cdot b$ implying the claim as $b = O(\frac{\ln{(1/\delta)}}{\Delta^2})$. 
\end{proof}
We prove the correctness of the algorithm in the next subsection. 

\subsection{The Analysis} 

The analysis consists of the following two main parts. 
Firstly, when we visit the most biased coin in the stream, it will defeat the  $\king$ with a large probability and become the next $\king$ itself.

\begin{lemma}\label{lem:main-pick}
	The probability that the most biased coin does \emph{not} defeat the $\king$  is at most $(\delta/2)$. 
\end{lemma}

Secondly, after the most biased coin become the $\king$, it will remain the $\king$ for the remainder of the stream with a large probability. 

\begin{lemma}\label{lem:main-keep}
	The probability that the most biased coin is ever defeated as the $\king$ is at most $(\delta/2)$. 
\end{lemma}

The proof of these key lemmas are postponed to the next two parts. \Cref{thm:main} now follows easily from these and \Cref{clm:main-sample}. 

\begin{proof}[Proof of \Cref{thm:main}]
	\Cref{clm:main-sample} ensures the bound on the sample complexity of the algorithm, and \Cref{lem:main-pick,lem:main-keep} together with a union bound
	ensure that with probability at least $1-\delta$, we return the most biased coin as the answer. 
\end{proof}

\subsection*{Proof of~\Cref{lem:main-pick}}
\begin{proof}
Let $\king$ be any coin other than the most biased coin and suppose the next arriving coin is the most biased one (denoted by $\coinstar$). We can write the probability that $\coinstar$ defeats $\king$ based on the different level $\ell$ of challenges done between the two as follows: 
\begin{align*}
	\Pr\paren{\textnormal{$\coinstar$ loses to $\king$}} &\leq \sum_{\ell=1}^{\infty} \Pr\paren{\textnormal{$\coinstar$ loses to $\king$ at level $\ell \mid \coin$ has not lost until level $\ell-1$}} \\
	&\leq \sum_{\ell=1}^{\infty} 2 \cdot \exp\paren{-\ln{(1/\delta) \cdot r_{\ell}}} \tag{by \Cref{lem:coin-comp} and $s_{\ell}$ number of samples done in level $\ell$} \\
	&< 2\delta \cdot \sum_{\ell=1}^{\infty} \exp\paren{-3^{\ell}} \tag{by definition of $r_\ell = 3^{\ell}$ and since $\ln{(1/\delta)} \cdot r_{\ell} \geq \ln{(1/\delta)} + r_\ell$}\\
	&< (\delta/2) \tag{as this series converges to $<1/10$}
\end{align*}
Since the budget  is finite, $\king$ will  lose to $\coinstar$ in finite time with probability $1-(\delta/2)$. 
\end{proof}

\subsection*{Proof of \Cref{lem:main-keep}}

We first need to set up some notation. 
Let $T \in [n]$ denote the time step at which the most biased coin arrives in the stream (i.e., $\coin_T$ is the most biased coin $\coinstar$). 
We define the following random variables $\set{X_{ij}}$ for $i,\ell \geq 1$ as the number of coin tosses when comparing $\king$ with $\coin_{T+i}$ at level $\ell$ of their challenge (note that index $i$ refers to the $i$-th coin
that arrives \emph{after} the most biased coin, not from the beginning of the stream):
\begin{align*}
	X_{i\ell} = 
	\begin{cases} 
	0 &\qquad  \textnormal{if the challenge of $\coinstar$ and $\coin_{T+i}$ did \emph{not} reach level $\ell$} \\
	s_\ell &\qquad   \textnormal{otherwise}
	\end{cases}.
\end{align*}
For any $i \geq 1$, we further define $X_i = \sum_{\ell = 1}^{\infty} X_{i \ell}$ which is the number of coin tosses when challenging $\coin_{T+i}$ with the $\king$. Finally, 
define $Y_i := \sum_{j=1}^{i} X_j$.  We prove that with probability $\geq 1-(\delta/2)$, 
\begin{align}
\textnormal{for every $i \geq 1$:} \qquad Y_i < i \cdot b. \label{eq:main-forall} 
\end{align}

This proves \Cref{lem:main-keep} since: (1) the \emph{total} number of samples from the time the $\coinstar$ is chosen as $\king$ till the $i$-th next coin arrives in the stream is  $Y_i$ and (2) the $\king$ receives 
$b \cdot i$ budget by the time we reach the $i$-th coin; hence, having $Y_i < i \cdot b$ \emph{for all} $i$ simultaneously, implies that the $\king$ never exhausted its budget and hence was not overthrown till the end of the stream. 

In proving \Cref{eq:main-forall}, working directly with random variables defined above is rather tricky (as it will become evident from our proof). Hence, we instead define the following random variables: 
\begin{align*}
	\set{X'_{i,\ell}}: \qquad &X'_{i\ell} = 
	\begin{cases} 
	0 &\qquad  \textnormal{if the challenge of $\coinstar$ and $\coin_{T+i}$ did \emph{not} reach level $\ell$} \\
	r_\ell &\qquad   \textnormal{otherwise}
	\end{cases}; \tag{the difference with $X_{i\ell}$ is that we are setting $X'_{i\ell}$ to $r_\ell$ not $s_\ell$}\\
	\set{X'_{i}}: \qquad &X'_i = \sum_{\ell = 2}^{\infty} X'_{i,\ell}, \qquad \set{Y'_{i}}: \qquad Y'_i = \sum_{j=1}^{i} X'_j  \tag{note that in defining $X'_i$ we are starting $\ell$ from $2$ and \emph{not} $1$}. 
\end{align*}
By these definitions, for every $i \geq 1$, 
\begin{align*}
X_i &\leq \paren{\frac{4}{\Delta^2} \cdot \ln{(1/\delta)}} \cdot X'_i + s_1,  \qquad Y_i \leq \paren{\frac{4}{\Delta^2} \cdot \ln{(1/\delta)}} \cdot Y'_i + i \cdot s_1.
\end{align*} 
Hence, by the choice of budget increment $b$, to prove \Cref{eq:main-forall}, it suffices to prove the following:
\begin{align}
	\textnormal{for every $i \geq 1$:} \qquad Y'_i < C \cdot i. \label{eq:main-forall-2}
\end{align}

We now prove \Cref{eq:main-forall-2}.  The approach is to bound the expected value of each $Y'_i$, prove that it is concentrated (by showing $X'_j$ is a sub-exponential variable and apply Bernstein's inequality to $Y'_i$), and show that this concentration is enough to do a union bound over a $\Theta(n)$-length stream.

\begin{claim}\label{clm:main-expect}
	For all $i > 0$, $\expect{Y'_i} \leq i$. 
\end{claim}
\begin{proof}
	We prove that $\EX[X'_j] \leq 1$ for every $j \in [i]$ which implies the claim by linearity of expectation. 
	For every level $\ell > 1$ of the challenge, we have, 
	\begin{align}
		\Pr\paren{\textnormal{challenge gets to level $\ell$}} &\leq \Pr\paren{\textnormal{challenge gets to level $\ell \mid$ challenge gets to level $\ell-1$}} \notag \\
		&\leq 2\exp\paren{-\ln{(1/\delta)} \cdot r_{\ell-1}}. \label{eq:main-prob}
	\end{align}
	where the inequality is by \Cref{lem:coin-comp} (for the event of $\coinstar$ losing) and $s_{\ell-1}$ number of samples done in level $\ell-1$. 
	For the random variable $X'_j$, we have, 
	\begin{align*}
		\expect{X'_j} &\leq \sum_{\ell > 1} \Pr\paren{\textnormal{challenge gets to level $\ell$}} \cdot r_{\ell}  \\
		 &\leq \sum_{\ell > 1} 2\exp\paren{-\ln{(1/\delta)} + r_{\ell-1}}\cdot r_{\ell} \tag{by \Cref{eq:main-prob} and since $\ln{(1/\delta)} \cdot r_{\ell} \geq \ln{(1/\delta)} + r_\ell$} \\
		&\leq 2\delta \cdot \sum_{\ell > 1} \frac{3^{\ell}}{\exp\paren{3^{\ell-1}}}\leq \delta \tag{$r_\ell = 3^{\ell}$, $r_{\ell-1} = 3^{\ell-1}$, and  series converges to $<1/2$}. 
	\end{align*}
	Noting that $\delta < 1$ concludes the proof of the claim. 
\end{proof}
\noindent
\Cref{clm:main-expect} suggests that $\set{Y'_i}$ behave as we require in \Cref{eq:main-forall} in expectation. 
To prove a concentration bound, we prove that each $(X'_i-\expect{X'_i})$ is a sub-exponential variable  with small $\kappa$ (see~\Cref{sec:concentration}).

\begin{claim}\label{clm:main-sub}
	For all $i > 0$, $(X'_i-\expect{X'_i})$ is a sub-exponential random variable with  $\kappa=\frac{15}{\ln{(1/\delta)}}$. 
\end{claim}
\begin{proof}
	Fix any $t > 0$ and let $\ell$ be the \emph{largest} level where $\sum_{j=2}^{\ell} r_{j} \leq t$. Note that since $\set{r_j}_{j=1}^{\infty}$ forms a geometric series, we have $t \leq 5 \cdot r_\ell$. We thus have, 
	\begin{align*}
		\Pr\paren{\card{X'_i-\expect{X'_i}} > t} &\leq \Pr\paren{\textnormal{challenge gets to level $\ell$}} \\
		&\leq 2\exp\paren{-\ln{(1/\delta)} \cdot r_{\ell-1}} \tag{by \Cref{eq:main-prob}} \\
		&\leq 2\exp\paren{-\ln{(1/\delta)} \cdot \frac{t}{15}}. \tag{as $t \leq 5 \cdot r_{\ell} = 15 \cdot r_{\ell-1}$} 
	\end{align*}
	This implies the proof by definition of sub-exponential variables in \Cref{eq:sub-exponential} of~\Cref{sec:concentration}. 
\end{proof}

We can now apply Bernstein's inequality (\Cref{prop:bernstein}) to $(Y'_i - \expect{Y'_i}) = \sum_{j=1}^{i} (X'_j - \EX[X'_j])$ (since by \Cref{clm:main-sub}, variables $\set{X'_j}$ are independent and sub-exponential with $\kappa=\frac{15}{\ln{(1/\delta)}}$): 
\begin{align*}
	\Pr\paren{Y'_i \geq C \cdot i} &\leq \Pr\paren{\card{Y'_i - \expect{Y'_i}} \geq (C-1) \cdot i} \leq 2 \cdot \exp\paren{-c \cdot \min\paren{\frac{(C-1)^2 \cdot i^2}{\kappa^2 \cdot i}, \frac{(C-1) \cdot i}{\kappa}}} \tag{by \Cref{clm:main-expect} to 
	bound the expectation and \Cref{prop:bernstein} where $c > 0$ is a constant} \\
	&\leq 2 \cdot \exp\paren{-c \cdot \frac{(C-1) \cdot i \cdot \ln{(1/\delta)}}{15}} \tag{by the value of $\kappa=\frac{15}{\ln{(1/\delta)}}$ in \Cref{clm:main-sub}} \\
	&\leq (\delta/2) \cdot \exp\paren{-i}. \tag{by picking $C$ to be a sufficiently large constant}
\end{align*}
Finally, by this and a union bound for all choices of $i$, we have,
\begin{align*}
	\Pr\paren{\exists i: Y'_i \geq C \cdot i} \leq (\delta/2) \cdot \sum_{i=1}^{n} \exp\paren{-i} < (\delta/2). \tag{as this series converges to $\frac{1}{e-1} < 1$}
\end{align*}
This proves that with probability $ \geq 1-(\delta/2)$, \Cref{eq:main-forall-2} holds,  finalizing the proof of \Cref{lem:main-keep}. 
\begin{remark}
The proof of \Cref{lem:main-keep} implies a bound of a random walk with flexible step size (rather than $-1$ and $+1$). As the analysis of such type of random walk may be useful in other settings as well, we 
abstract out this problem in \Cref{sec:rand-walk} and analyze it directly. 
\end{remark}

%% file: top-k.tex
\newcommand{\algtopk}{\ensuremath{\textnormal{\textsc{Federated-Game-of-Coin}}}\xspace}

\newcommand{\Kings}{\ensuremath{\textnormal{\textsf{KINGS}}}\xspace}

\newcommand{\pcoin}{\ensuremath{\overline{\coin}}\xspace}
\newcommand{\pking}{\ensuremath{\overline{\king}}\xspace}

\newcommand{\defeat}{\ensuremath{{D}}\xspace}

\newcommand{\Ntr}{\ensuremath{N_{\textnormal{trial}}}}

\newcommand{\pevent}{\ensuremath{\event_{\textnormal{pivot}}}}

\section{Top $k$ Most Biased Coins: An $O(k)$-Coin Memory Algorithm}\label{sec:top-k}

We now consider the more general problem of finding the $k$ most biased coin for any integer $k \geq 1$. In this problem, 
we have a collection of coins $\set{\coin_i}_{i=1}^{n}$ arriving in a stream; for simplicity of notation, we use $\coin_{[i]}$ to denote the $i$-th most biased coin among these. 
Our goal is then to find the $k$ coins with largest biases, namely, $\set{\coin_{[1]},\ldots,\coin_{[k]}}$ (in no particular order) for a given integer $k \geq 1$. The gap parameter for this problem, denoted by $\Delta_k$,
is now defined as the gap between the bias of the $k$-th most biased coin and $(k+1)$-th one, namely $\coin_{[k]}$ and $\coin_{[k+1]}$. 

We present a streaming algorithm for this problem with asymptotically optimal space complexity as well as sample complexity (by the lower bound of~\cite{KalyanakrishnanTAS12}). 

\begin{theorem}[Formalization of~\Cref{res:main2}]\label{thm:top-k}
	There exists a streaming algorithm that given an integer $k \geq 1$, $n$ coins arriving in a stream with gap parameter $\Delta_k$ (between $k$-th and $(k+1)$-th most biased coins) and confidence parameter $\delta \in (0,1/2)$, finds the $k$ most biased coins with 
	probability at least $1-\delta$ using $O(\frac{n}{\Delta^2} \cdot \log{(k/\delta)})$ coin tosses and a memory of $O(k)$ coins.  
\end{theorem}

\subsection{High Level Overview} 

We follow the same  ``budgeting'' strategy as our algorithm in \Cref{sec:top-k}. However, as stated  in~\Cref{sec:intro}, there are two main challenges that we need to address: (1) 
we now need to maintain $k$  ``kings'', namely, $\Kings = \set{\king_1,\ldots,\king_k}$ but can no longer compare each arriving coin with (or assign
a unit of budget to) every king  (otherwise there will be $\Omega(nk)$ samples); and (2) we need to collect all the top-$k$ coins and still cannot guarantee any suitable (probabilistic) outcome while comparing any two of 
these coins to each other (as there may be no gap between their biases). 

There is a natural way for addressing the first challenge: instead of comparing each arriving coin with the $k$ $\king$-coins using $O(k)$ coin tosses, {delay} processing of arriving coins, by storing them in a \emph{buffer} $B$, until we collect roughly $k$ of them; then handle all these coins using $O(k\log{k})$ coin tosses in total by running the following \emph{trial}: pick a \emph{pivot} coin from $B$, compare this pivot with every $\king$ and every coin in $B$, and 
\emph{prune} the buffer by discarding any coin with empirical bias less than the pivot in this trial. 
Assuming we prune a \emph{constant fraction} of the buffer per each trial (which seems doable, at least in \emph{expectation}, by picking the pivot \emph{randomly}), we can spend $O(k\log{k})$ coin tosses per trial and sample $O(n\log{k})$ coins in total. Finally, to compare a $\king$ with a pivot, we can use the challenge subroutine (in our algorithm in~\Cref{sec:main}): allow
any $\king$  to use its budget and only consider it lost in a challenge when it exhausts its budget entirely (the coins in the buffer will not collect any budget). 
We can also allocate $O(k\log{k})$ budget per each trial (and \emph{not} per each arriving coin) 
and hope that this should allow us, similar to~\Cref{sec:main}, to argue that any top-$k$ pivot will win against any non-top-$k$ $\king$ and will later remain in $\Kings$ till the end.  

Except that this actually would \emph{no longer work}, which brings us to the second 
(and the main) challenge raised above. The problem with the above reasoning is that it does not take into account the outcome of challenging a top-$k$ coin as a pivot with another top-$k$ coin as a $\king$. 
In such a challenge, the previous probabilistic guarantees in~\Cref{sec:main} no longer hold as we have no control on the gap between the biases of these coins. 
For instance, it is entirely possible that a top-$k$ pivot completely depletes the budget of a top-$k$ $\king$ and the troublesome part is 
that this is the same exact behavior we would also except from a top-$k$ pivot when challenging a non-top-$k$ $\king$ (with no apparent way of distinguishing between the two cases). 
At the same time, it is also completely possible that the bias of two top-$k$ coins is almost the same
and hence their challenges would be completely noisy. The choice of a top-$k$ pivot also highlights another problem: we need to be very ``cautious'' in the pruning step as when choosing a top-$k$ pivot, we may inadvertently discard
other top-$k$ coins (either in the buffer or among $\Kings$) when they lose to this top-$k$ coin -- note that this goes exactly opposite of our goal of pruning a constant factor of the buffer per each trial. 

We address the latter challenge by relaxing the requirement of the algorithm (and the analysis) in maintaining the top-$k$ coins among $\Kings$ throughout the entire length of the stream (after their arrival). 
In other words, in the course of our algorithm, the top-$k$ coins may float between $\Kings$ (and having a budget) and the buffer $B$ (with no budget). This in turn requires us to relax 
our pruning rule so that the top-$k$ coins in the buffer do not get discarded in a trial: this is done by limiting the cases when a discard can happen (for instance not doing any pruning when the pivot joins the $\Kings$), 
while still ensuring the constant fraction pruning (in expectation) per trial. Finally, the analysis now needs to take into account that a top-$k$ coin may \emph{repeatedly} exhausts its budget and there will 
be periods of trials in the stream when a top-$k$ coin resides in $B$ with no budget (which we refer to as \emph{risky} trials). Fortunately, by modifying the algorithm appropriately, we can limit the \emph{length} and the \emph{frequency} of such 
periods throughout the stream and show that with high (constant) probability, any top-$k$ coin will indeed remain among $\Kings \cup B$ till the end.

\subsection{The Algorithm}

We now present our algorithm in this section. The input to our algorithm is a set of $n$ coins $\set{\coin_i}_{i=1}^{n}$ arriving in an arbitrary order in the stream, the gap parameter $\Delta_k$ (the gap between the bias of $\coin_{[k]}$ and $\coin_{[k+1]}$), 
and the confidence parameter $\delta \in (0,1)$ (the algorithm does not need to know the value of $n$ in advance). We use the following parameters: 
\begin{align*}
	\qquad \qquad \qquad \qquad \qquad &\!\!\!\set{r_\ell}_{\ell=1}^{\infty}: \quad r_\ell = 3^{\ell};
	\\ \tag{intermediate parameters to define the number of samples per each level of the challenge} \\
	&\!\!\!\set{s_\ell}_{\ell=1}^{\infty}: \quad s_\ell :=  16\cdot\frac{4}{\Delta_{k}^2} \cdot \ln{(k/\delta)} \cdot r_\ell;   \\ \tag{the number of coin tosses per each level of the challenge} \\
	&b := 16\cdot\frac{4}{\Delta_{k}^2} \cdot C \cdot \ln{(k/\delta)} + s_1; \tag{the budget given to each $\king$ once the buffer is full} \\
	&K := 10 \cdot k. \tag{the limit on the size of the buffer}
\end{align*}
\noindent
And our algorithm can be presented as follows: 

\begin{tbox}
	\textbf{Algorithm} $\algtopk$: 
	\begin{enumerate}[label=($\arabic*$)]
		\item\label{linefillking} Initialize $\Kings = \set{\king_1,\ldots,\king_k}$ by the first $k$ arriving coins and let $B$ be the buffer.  
		\item For any $\king_i \in \Kings$, define the budget $\Budget_i := \budget{\king_i}$ which is initialized to $0$. 
		\item\label{linefillbuffer} While number of coins in $B$ is less than $K$, add the next coin in the stream to $B$. 
		\item \textbf{\underline{Trial subroutine:}} Otherwise, run the following \emph{trial}: 
		\begin{enumerate}[label=($\alph*$)] 
			\item\label{linepickpivot} Pick a \emph{pivot} $\pcoin$ uniformly at random from $B$. Increase the budget $\Budget_i$ of $\king_i$ by $b$. 
			\item \textbf{\underline{Buffer-challenge:}} For each $\coin_{i} \in B$: sample both $\coin_i$ and $\pcoin$ for $s_1$ times and record which one had a higher empirical bias. 
			\item \textbf{\underline{King-challenge:}} For each $\king_{i} \in \Kings$: run the \textbf{challenge subroutine} of $\mainalg$ between $\pcoin$ and $\king_i$ (with new
			 $\set{s_\ell}$ and budget $\budget{\king_i}$) and record which coin won the 
			challenge (but do not discard any coin). 
		    \item Let $\defeat$ denote the recorded number of times $\pcoin$ was defeated in the trial. 
		    \item \emph{Discard case:} If $\defeat \geq k$: discard $\pcoin$, any coin in $\Kings \cup B$ that lost to $\pcoin$. 
		    Then fill up the remainder of $\Kings$ with the  arriving coins of the stream and go to \ref{linefillbuffer}.
		    \item \emph{Swap case:} If $\defeat < k$: pick $\pking$ uniformly at random coins in $\Kings$ that were defeated by $\pcoin$ (such a coin should exists) 
		    and swap $\pcoin$ and $\pking$, i.e., make $\pcoin$ a new king (with zero budget) and add $\pking$ to $B$. 
		    Then repeat the trial by going to~\ref{linepickpivot}. 
		\end{enumerate}
		\item\label{linefinalsample} At the end, sample each of the coins in $\Kings \cup B$ for $s_{1}$ times and return the top-$k$ ones according to their empirical bias as the answer. 
	\end{enumerate}
\end{tbox}
\noindent
This concludes the description of the algorithm. We note that at this point, the bound on the sample complexity of this algorithm is \emph{in expectation} and not deterministically. For simplicity of exposition, we analyze this variant of the algorithm first
and then point out, in~\Cref{rmk:trial-high-prob}, how to change this slightly so that the algorithm \emph{never} (deterministically) uses more than a fixed certain number of coin tosses bounded by $O(\frac{n}{\Delta_k^2} \cdot \log{(k/\delta)})$ (this extension is straightforward). 
We present the analysis of the algorithm in the next section. 

\subsection{The Analysis} 

There are two main parts in the analysis. Unlike our $\mainalg$ algorithms, bounding the sample complexity of this new algorithm is not straightforward and requires a careful analysis which is the subject of the following lemma. 
\begin{lemma}\label{lem:topk-sample}
	The expected number of coin tosses by the algorithm is $O(\frac{n}{\Delta_{k}^{2}}\cdot\log(\frac{k}{\delta}))$. 
\end{lemma}

The main part however as before is to prove the correctness of the algorithm, which is done by the following lemma. 

\begin{lemma}\label{lem:topk-correct}
	The probability that even a single $\coin_{[j]}$ for $j \in [k]$ is discarded before the end of the stream (before~Line~\ref{linefinalsample}) is is at most $\frac{\delta}{2}$.
\end{lemma}

In the following, we first prove each of these two lemmas and then show that how~\Cref{thm:top-k} follows easily from these results. 

\subsection*{Proof of \Cref{lem:topk-sample} (Sample Complexity)}

Let us recall that in the algorithm, coin tossing happens only during a \emph{trial} in the trial subroutine (ignoring the last $O(k \cdot s_1)$ samples in~\ref{linefinalsample} which are clearly within the desired sbounds on sample complexity
 by definition of $s_1$), namely, when the buffer is full and we pick a pivot for challenging the other coins. Let $\Ntr$ denote the \emph{number of trials} in the algorithm. We have the following claim based on a similar amortized analysis
 as in our $\mainalg$ algorithm. 
 
\begin{claim}\label{clm:comp-with-trial}
The total number of coin tosses in the algorithm is $O(k\cdot b\cdot \Ntr)$.
\end{claim}

\begin{proof}
In each trial, each $\king$ will be given a budget of $b$ and so the total budget given to all kings throughout the algorithm is $(k\cdot b\cdot \Ntr)$. This ensures that the total number of coin tosses in king-challenges is at most $O(k \cdot b \cdot \Ntr)$.
Moreover, during buffer-challenges, any coin in the buffer will also be tossed $s_1$ times if it is not the pivot and $K \cdot s_1$ times if it is the pivot. Since $K = O(k)$, and $s_1 = O(b)$, we obtain that the total number of coin tosses in buffer-challenges
is also $O(k\cdot b\cdot \Ntr)$, finalizing the proof. 
\end{proof}

\Cref{clm:comp-with-trial} implies that we can bound the sample complexity of the algorithm by bounding $\Ntr$ which is the content of the next claim. 

\begin{claim}\label{clm:trial-num}
The expected number of trials is $\expect{\Ntr} = O(\frac{n}{k})$.
\end{claim}
\begin{proof}

Consider the following event: 
\begin{itemize}
	\item $\pevent$: the pivot coin $\pcoin$ loses to at least $k$ coins and wins over at least $k$ other coins. 
\end{itemize}
Whenever $\pevent$ happens, we discard at least $k$ coins from the buffer. By lower bounding the probability of this event by a constant, we can then argue that the expected number of coins discarded in each trial is $\Omega(k)$. As the next trial 
can only happen when the buffer again becomes full (thus after $\Omega(k)$ new coins are visited), this will allow us to argue that the expected number of trials before we process the entire stream is $O(n/k)$. 

We now lower bound the probability that $\pevent$ happens by considering a simpler case that ensures $\pevent$. 
The total number of coins in $\Kings \cup B$ is $K+k = 11k$. Let us sort these coins in decreasing order of their biases as $\coin_{(1)}, \coin_{(2)}, \ldots, \coin_{(K+k)}$. We further partition these 
coins into the \emph{top} part $Top := \set{\coin_{(1)},\ldots,\coin_{(5k)}}$, the \emph{middle} part $Mid := \set{\coin_{(5k+1)},\ldots,\coin_{(7k)}}$, and the \emph{bottom} part $Bot := \set{\coin_{(7k+1)}, \cdots, \coin_{(11k)}}$. See~\Cref{fig:strongpivot} 
for an illustration. 

Now firstly note that since we only have $k$ coins, the probability that the pivot is chosen from $Mid$ is at least $\frac{2k-k}{K} = \frac{1}{10}$. In the following, we condition on this event. Note that conditioned on this event, 
any coin in $Top$ would lose to $\pcoin$ with probability at most $1/2$, and any coin in $Bot$ which is \emph{not} in $\Kings$ would win against $\pcoin$ with probability at most $1/2$ (a coin in $Bot$ which is a $\king$ may have collected a lot of budget and thus
still have a more chance of winning against $\pcoin$ even though its bias is less than it). We define the following random variables. 

\begin{itemize}
\item $X_{lose}$: number of coins in $Top$ that lose to $\pcoin$ -- let $X_{win} = \card{Top} - X_{lose}$. 
\item $Y_{win}$: number of coins in $Bot \setminus \Kings$ that win against $\pcoin$ -- let $Y_{lose} = \card{Bot \setminus \Kings} - Y_{win}$. 
\end{itemize}
We thus have $\expect{X_{lose} \mid \pcoin \in Mid} \leq 5k/2$ and $\expect{Y_{win} \mid \pcoin \in Mid} \leq 3k/2$. As such,
\begin{align*}
	\Pr\paren{X_{win} < k \mid \pcoin \in Mid} \leq \Pr\paren{X_{lose} \geq 4k \mid \pcoin \in Mid} \leq \frac{5}{8}, \\
	\Pr\paren{Y_{lose} < k \mid \pcoin \in Mid} \leq \Pr\paren{Y_{win} \geq 2k \mid \pcoin \in Mid} \leq \frac{3}{4},	
\end{align*}
where both inequalities are by Markov bound. Moreover, we have, 
\begin{align*}
	\Pr\paren{X_{win} \geq k \wedge Y_{lose} \geq k \mid \pcoin \in Mid} \leq \paren{1-\frac{5}{8}} \cdot \paren{1-\frac{3}{4}} = \frac{3}{32},
\end{align*}
since these events are independent of each other. However, notice that whenever the event above happens, we would be in the `discard case' of the algorithm (since $\pcoin$ has lost to at least $k$ coins in $Top$) 
and we would discard at least $k$ coins (all the coins in $Y_{lose}$ that belong to $Bot$). Hence, 
\begin{align*}
	\Pr\paren{\pevent} &\geq \Pr\paren{X_{win} \geq k \wedge Y_{lose} \geq k \wedge \pcoin \in Mid}  \\
	&= \Pr\paren{\pcoin \in Mid} \cdot \Pr\paren{X_{win} \geq k \wedge Y_{lose} \geq k \mid \pcoin \in Mid} \\
	&\geq \frac{1}{10} \cdot \frac{3}{24} = \frac{3}{320}  \geq \frac{1}{200}.  
\end{align*}

\begin{figure}[t!]
\centering
\includegraphics[width=1.0\textwidth]{strongpivot.png}
\caption{\label{fig:strongpivot}When picking the pivot between the $(5k+1)$-th and $(7k)$-th most biased coins (namely, from $Mid$) from the current $\Kings$ and buffer, the probability for the \emph{pivot} to win over at least $k$ coins \emph{and} lose against
at least $k$ coins is at least a constant.}
\end{figure}

This implies that the expected number of coins that are discarded in each trial is at least $k/100$. Moreover, note that this lower bound holds in every trial \emph{independent} of the outcome of the past trials (event hough the events between the two trials may not necessarily be independent). This means that the distribution of $\Ntr$ stochastically dominates the distribution of number of times we see a head by tossing a biased coin with probability $1/200$ of showing a head. For the latter distribution we know that 
the expected number of tries before we see $n$ heads is $200 \cdot n/k$ and hence we also have $\expect{\Ntr} \leq 200 \cdot n/k$ (as after seeing $n$ coins the trials are finished). 
\end{proof}

We now formally conclude the proof of \Cref{lem:topk-sample}. Combine \Cref{clm:comp-with-trial} and \Cref{clm:trial-num}, one can observe that the expected number of coin tosses will be $O(k\cdot b\cdot N_{tr}) = O(k \cdot b \cdot \frac{n}{k}) = O(b\cdot n)$. And according to the definition, this is $O(\frac{n}{\Delta_{k}^{2}}\cdot\log(\frac{k}{\delta}))$ as desired.

\begin{remark}\label{rmk:trial-loss-bound}
The lower bound of $\frac{1}{200}$ on the probability of $\pevent$ proved in~\Cref{clm:trial-num} is quite loose and is easy to see several ways of improving it. However, since this bound is already enough for our purpose and in the interest of simplifying the proof, we opted to use this simple argument anyway.  
\end{remark}

\begin{remark}\label{rmk:trial-high-prob}
We remark that the probability that $\Ntr$ is more than twice its expectation is exponentially small in $\Theta(n/k)$ (which we can assume $k$ is at most $\sqrt{n}$ since whenever $k \geq \sqrt{n}$, we can simply toss each coin $O(\log{n})$ times to obtain its `almost true' bias and still be within the correct budget -- but in this case, we can simply run a deterministic algorithm for finding top-$k$ coins in the stream over the empirical biases). As such, 
we can simply modify the algorithm by terminating with an arbitrary answer whenever the $\Ntr$ reaches twice its expected value -- this can only decrease the probability of success by $\exp\paren{-\Theta(\sqrt{n})}$ (which again can be assumed to be 
always $o(\delta)$ by a similar argument as why assuming $k \leq \sqrt{n}$ is without loss of generality). This means that the sample complexity of our algorithm can be bounded \emph{deterministically} also.  
\end{remark}

\subsection*{Proof of \Cref{lem:topk-correct} (Correctness of the Algorithm)}

Let us start by giving some intuition about the proof before diving into the technical details. The ideal scenario for the algorithm is if we start with all the top-$k$ coins appearing at the beginning of the stream and so from the get go, they all belong 
to $\Kings$. In such a scenario, we can invoke~\Cref{lem:main-keep} from~\Cref{sec:main} in an almost black-box way and argue that the budgeting scheme allows for all these coins to remain in $\King$ till the very end of the stream
with probability at least $1-\delta$. The reason this works is that in this case, we never need to consider comparing two top-$k$ coins with each other (as the pivots are sampled from the buffer alone). 

Of course, in general, we will not have all top-$k$ coins as $\Kings$ in the beginning. The first thing we need
 to worry is when a top-$k$ coin enters the buffer (and for now let us assume there is no other top-$k$ coin the buffer for the next foreseeable streaming steps): since this top-$k$ coin does not have any budget, can we still hope to have
it around for multiple trials before it is chosen as the pivot and even have a chance of joining the $\Kings$? Since the pivot is chosen uniformly at random,  we would expect this top-$k$ coin to become 
a pivot itself within the next $O(k)$ trials. Thus, we only need this coin to remain in the buffer for the next $O(k)$ trials; as the coins are sampled $O(\log{k})$ times in each trial, we can guarantee this event. 
Moreover, once this coin is chosen as the pivot, we can also guarantee that it will join the $\Kings$ by the same argument as \Cref{lem:main-pick} in~\Cref{sec:main}. 

Already at this point, we encounter a problem: What if this top-$k$ coin swaps one of the top-$k$ coins in $\Kings$? Indeed, our pruning rule allows us to argue that with high probability we will \emph{not} have a \emph{discard} step when 
this coin joins $\Kings$ but inevitably a swap needs to happen and we may very well swap a top-$k$ coin with with another top-$k$ coin. This can become even more challenging when multiple top-$k$ coins all join the buffer. 

Our main argument here is to show that it is possible to partition the execution of the algorithm over the stream into \emph{long} sequences of ``relative safety'' in which
no top-$k$ coin belongs to the buffer and the top-$k$ coins in $\Kings$ start to accumulate budget (which allows us to do union bound over these long sequences), and \emph{short} outbursts of ``risky'' trials in which the budget 
of every $\king$ may be depleted and the only thing that saves us through these risky trials is that their numbers  are small (so we can directly use a union bound over them). The final step is to use a simple potential function argument to prove that the total number of such risky outbursts is small and most of the stream involves the long non-risky trials (so even though the budgets of the top-$k$ coins in $\Kings$ may get restarted after each risky outbursts, we can still expect them to survive all these outbursts and not
get discarded by the end of the stream). We now formalize this intuition. 

\medskip

We start by setting up our notation. Let us define: 
\begin{itemize}
\item \emph{Risky trial}: A trial with at least one of the top-$k$ coins present in the buffer $B$;
\item \emph{Non-risky trial}: A trial without any $\coin$ from the top-$k$ coins present in the buffer $B$;
\item \emph{(Non-risky) Chunk}: A maximal sequence of consecutive non-risky trials. 
\end{itemize}

What we intend to prove is as follows (see~\Cref{fig:trialandchunk} for an illustration of these definitions): 

\begin{enumerate}[label=(\roman*)]
\item During any single non-risky chunks, with large probability of $1-\poly(\delta)/\poly(k)$, we will not encounter any `swap case' or `discard case' of the algorithm that \emph{removes a top-$k$ coin from $\Kings$}.  
In other words, we only enter a risky trial on the condition of a new arriving top-$k$ coin joining the buffer from the stream (\Cref{clm:non-risky-chunks}).
\item For a single risky trial, with large probability of $1-\poly(\delta)/\poly(k)$, no $\coin$ among the top-$k$ will be discarded, even though we may encounter many `swap case' or `discard case' in the algorithm (\Cref{clm:risky-trial}).
\item The expected number of risky trials as well as (non-risky) chunks is $\text{poly}(k)/\text{poly}(\delta)$ where the bound is small enough to do a union bound over all occurrences of the above  cases (\Cref{clm:trial-chunk-num}).
\end{enumerate}
 
 \begin{figure}[h!]
\centering
\includegraphics[width=1.0\textwidth]{trialsandchunks.png}
\caption{\label{fig:trialandchunk} An illustration of the notation, events and arguments adopted in the proof of \Cref{lem:topk-correct}.}
\end{figure}

The proof of the following claim is analogous to~\Cref{lem:main-keep}. 
\begin{claim}\label{clm:non-risky-chunks}
With probability at least $1-\frac{\delta^{2}}{k^{15}}$, any top-$k$ coin in $\Kings$ will not be defeated during a fixed (non-risky) chunk.  
\end{claim}
\begin{proof}
Let $T$ denote the time step at which the first non-risky trial starts after a bunch of risky trials. We define the following random variables $\set{X_{m,i,\ell}}$ for $m,i,\ell \geq 1$ as the number of coin tosses when comparing the $m$-th $\king$ ($\king_{m}$) with the pivot $\overline{\coin}$ (which is not among the top-$k$ coins since this is a non-risky trial) on the $\ell$-th level of the $i$-th trial after $T$.  We define:
\begin{align*}
	X_{m,i,\ell} = 
	\begin{cases} 
	0 &\qquad  \textnormal{if the challenge of $\coin_{\king_{m}}$ and $\overline{\coin}$ of the $i$-th trial did \emph{not} reach level $\ell$} \\
	s_\ell &\qquad   \textnormal{otherwise}
	\end{cases}.
\end{align*}
And similarly, we define $X_{m,i} = \sum_{\ell = 1}^{\infty} X_{m, i, \ell}$ and $Y_{m,i} := \sum_{j=1}^{i} X_{m,j}$. Now, instead of proving a $1-(\delta/2)$ probability, we prove that with probability at least $1-(\frac{\delta^{2}}{k^{16}})$: 
\begin{align*}
\textnormal{for every $i \geq 1$:} \qquad Y_i < i \cdot b.
\end{align*}
Without repeating too much the technical details of the proof of~\Cref{lem:main-keep}, we can define $X^{'}_{m,i,\ell}$, $X^{'}_{m,i}$ and $Y^{'}_{m,i}$ as we did exactly in that proof. Then we can replace the $-\ln(1/\delta)$ term in \Cref{clm:main-expect} with $-16\cdot\ln(k/\delta)$. The bound will therefore become:
\begin{align*}
	\expect{X'_{m,j}} &\leq \sum_{\ell > 1} \Pr\paren{\textnormal{challenge gets to level $\ell$}} \cdot r_{\ell}  \\
	 &\leq \sum_{\ell > 1}\exp\paren{-16\cdot\ln{(k/\delta)} + r_{\ell-1}}\cdot r_{\ell} \\
	&\leq 2\cdot\frac{\delta^{16}}{k^{16}} \cdot \sum_{\ell > 1} \frac{3^{\ell}}{\exp\paren{3^{\ell-1}}} \\
	&\leq \frac{\delta^{2}}{k^{16}}. \tag{as this series converges to $<1/2$} 
\end{align*}
Also, similar to the proof of \Cref{clm:main-sub}, we can show that for all $i$, $(X'_{m,i}-\expect{X'_{m,i}})$ is a sub-exponential random variable with  $\kappa=\frac{15}{16\cdot\ln{(k/\delta)}}$ by showing:
\begin{equation*}
	\Pr\paren{\card{X'_{m,i}-\expect{X'_{m,i}}} > t} \leq 2\exp\paren{-16\cdot\ln{(k/\delta)} \cdot \frac{t}{15}}.
\end{equation*}
Thus, by applying the same argument, we can show that:
\begin{align*}
\Pr\paren{Y'_{m,i} \geq C \cdot i} \leq \frac{\delta^{2}}{k^{16}}\cdot\exp(-i).
\end{align*}
Applying a union bound over all upcoming trials, and using the geometric series above, we can show that the probability of $\paren{\exists i: Y'_{m,i} \geq C \cdot i}$ is at most $\frac{\delta^{2}}{k^{16}}$.\par
Now notice that unlike the original proof in \Cref{lem:main-keep}, here the conclusion only applies to one $\king$. Thus, we need to apply another union bound. The number of $\king$s among the top-$k$ $\coin$s is at most $k$. Therefore, the probability of $\paren{\exists m,i: Y'_{m,i} \geq C \cdot i}$ should be at most $\sum_{m=1}^{k}\Pr\paren{\exists i: Y'_{m,i} \geq C \cdot i}\leq \frac{\delta^{2}}{k^{15}}$, finalizing the proof.
\end{proof}

\newcommand{\eventone}{\ensuremath{\event_{\textnormal{defeated-top}}}}
\newcommand{\eventtwo}{\ensuremath{\event_{\textnormal{pivot-top}}}}

\begin{claim}\label{clm:risky-trial}
With probability at least $1-\frac{\delta^{16}}{k^{9}}$, in a single risky trial, no top-$k$ coin will be discarded.
\end{claim}
\begin{proof}

We first argue that the only way for any top-$k$ coin to get discarded is if one of the following two events happens: 
\begin{itemize}
\item $\eventone$: $\overline{\coin}$ is \emph{not} a top-$k$ coin and defeats a top-$k$ coin in $\Kings \cup B$.  
\item $\eventtwo$: $\overline{\coin}$ is a top-$k$ coin that loses at least $k$ times (namely, have $D \geq k$). 
\end{itemize}
This is the case because of the following: if $\pcoin$ is not a top-$k$ coin, the only way for it to be able to discard a top-$k$ coin is if it wins against it (which is captured by$\eventone$). On the other hand, 
if $\pcoin$ is a top-$k$ coin, the only for it to be able to discard \emph{any} other coin, is if it enters a `discard case' that only happens if it loses at least $k$ times  (which is captured by $\eventtwo$). 

We now bound the probability of each of these two events. Fix any top-$k$ $\coinstar \in \Kings \cup B$ and the pivot $\pcoin$. Note that $\pcoin$ and $\coinstar$ are tossed \emph{at least} $s_1$ times
before we decide which one is the winner ($\coinstar$ may have a budget if it belongs to $\Kings$ on top of the $b \geq s_1$ provided to it at the beginning of this trial but we may and will ignore that for this argument). We have, 
\begin{align}
	\Pr\paren{\textnormal{$\coinstar$ loses to $\pcoin$}} &\leq 2 \cdot \exp\paren{-16\cdot \ln{(k/\delta)} \cdot r_1} \tag{by~\Cref{lem:coin-comp} and the choice of $s_1$} \\
	&\leq \frac{2\delta^{16}}{k^{16}}. \label{eq:may-need-later}
\end{align}
Doing a union bound over the at most $k$ choices of $\coinstar$, we have (as $k\geq 2$)
\begin{align*}
	\Pr\paren{\eventone} \leq \frac{\delta^{16}}{k^{14}}. 
\end{align*}

Let us now consider the case when $\pcoin$ is a top-$k$ coin. Let $\coin_i \in B$ be any coin which is \emph{not} a top-$k$ coin itself. By~\eqref{eq:may-need-later} (by now replacing the role of $\coinstar$ with $\pcoin$ and the previous $\pcoin$ with $\coin_i$), 
we have, 
\begin{align*}
	\Pr\paren{\textnormal{$\pcoin$ loses to $\coin_i$}} \leq \frac{2\delta^{16}}{k^{16}}.
\end{align*}
The trickier part is when we should compare $\pcoin$ with some $\king_i \in \Kings$ which is not a top-$k$ coin itself. Here, we can no longer ignore the fact that $\king_i$ may have collected some budget. So $\pcoin$ needs to win
 against $\king_i$ despite $\king_i$ having 
some budget (that we cannot necessarily bound beyond saying it is finite). However, we already proved an analogous statement like this in~\Cref{lem:main-pick} and the argument here is identical to that. Indeed, we have, 
\begin{align*}
	\Pr\paren{\textnormal{$\pcoin$ loses to $\king_i$}} &\leq \sum_{\ell=1}^{\infty} \Pr\paren{\textnormal{$\pcoin$ loses to $\king_i$ at level $\ell \mid \pcoin$ has not lost until level $\ell-1$}} \\
	&\leq \sum_{\ell=1}^{\infty} 2 \cdot \exp\paren{-16 \cdot \ln{(k/\delta) \cdot r_{\ell}}} \tag{by \Cref{lem:coin-comp} and $s_{\ell}$ number of samples done in level $\ell$} \\
	&< 2 \cdot \frac{\delta^{16}}{k^{16}} \cdot \sum_{\ell=1}^{\infty} \exp\paren{-3^{\ell}} \tag{by definition of $r_\ell = 3^{\ell}$ and since $\ln{(k/\delta)} \cdot r_{\ell} \geq \ln{(k/\delta)} + r_\ell$}\\
	&< \frac{\delta^{16}}{k^{16}} \tag{as this series converges to $<1/10$}.
\end{align*} 
By a union bound over the at most $11k$ non-top-$k$ coins in $\Kings \cup B$, we have that, (note that there are $<k$ top-$k$ coins other than $\pcoin$ and so for $\pcoin$ to lose to at least $k$ coins it should lose to some non-top-$k$ coins
and this union bound takes care of that)
\begin{align*}
	\Pr\paren{\eventtwo} \leq \frac{\delta^{16}}{k^{10}}. 
\end{align*}

A union bound over these two events (and a very loose upper bound) finalizes the proof. 
\end{proof}

\begin{claim}\label{clm:trial-chunk-num}
Assuming the events of~\Cref{clm:risky-trial} for every upcoming risky trial, with probability at least $1-\frac{\delta}{k^{3}}$, the number of risky trials is at most $10 \cdot \frac{k^{6}}{\delta}$. 
\end{claim}
\begin{proof}
Let us fix any risky trial. By definition, there must exists at least one top-$k$ coin, denoted by $\coinstar$, in the buffer $B$. By the random choice of the pivot, we will pick $\coinstar$ as the pivot 
with probability $\frac{1}{10k}$. Let us condition on this event. 

Moreover, note that since not all $\Kings$ are top-$k$ coins, there exists at least one non-top-$k$ king, denoted by $\king_i$, in $\Kings$. By conditioning on the event of~\Cref{clm:risky-trial}, 
$\coinstar$ will beat $\king_i$ and also enters a `swap case'. However, there is no guarantee that $\coinstar$ did not win against some other coins, some of which may actually be top-$k$ coin themselves. 
In that case, one of those may get swapped with $\coinstar$ instead of $\king_i$. Still, considering we pick $\pking$ to swap with $\coinstar$ uniformly at random, there is at least a $\frac{1}{k}$ chance that we pick
to swap $\king_i$ with $\coinstar$.  This means that, assuming the event in~\Cref{clm:risky-trial}, with probability at least $\frac{1}{10k^2}$, we will swap $\coinstar$ with $\king_i$. 

An important observation here is that as long as the event in~\Cref{clm:risky-trial} continues to happen, we will never \emph{decrease} the number of top-$k$ coins in $\Kings$ (this actually follows from the event $\eventone$ bounded 
in~\Cref{clm:risky-trial} and not the exact statement of the claim itself). This, plus the above fact implies that in each trial, we have a probability of $\geq \frac{1}{10k^2}$ to increase the number of top-$k$ 
coins among the $\Kings$ (and we will not decrease it conditioned on~\Cref{clm:risky-trial}). As the number of top-$k$ coins in $\Kings$ can be increased to $k$ only, we can conclude that the expected 
number of risky trials before we increase the top-$k$ coins in $\Kings$ to become $k$ is $10k^3$. Hence, by Markov bound, with probability $1-\frac{\delta}{k^3}$, we can only have $10 \cdot \frac{k^6}{\delta}$ 
risky trials. 
\end{proof}

We can now use~\Cref{clm:risky-trial} and~\Cref{clm:trial-chunk-num} and do a union bound (step by step on each upcoming risky trial) to argue that: the number of risky trials is $10 \cdot \frac{k^{6}}{\delta}$ and in each one, 
we will only lose a top-$k$ coin with probability at most $\frac{\delta^{16}}{k^9}$; hence, with probability 
\begin{align*}
1-\paren{\frac{\delta}{k^3}+\frac{10k^6}{\delta} \cdot \frac{\delta^{16}}{k^9}} \geq 1-\frac{2\delta}{k^3} \tag{as $\delta \leq <1/2$}
\end{align*} 
we will keep all the top-$k$ coins throughout all the risky trials and will not have more than $10 \cdot \frac{k^{6}}{\delta}$ risky trials. 

Furthermore, since between any two (non-risky) there should be a risky trial, the above bound gives us an upper bound of $10 \cdot \frac{k^{6}}{\delta}$ on the number of (non-risky) chunks as well. Thus, by applying~\Cref{clm:non-risky-chunks} 
and a union bound over all these chunks, we obtain that with probability $1-\frac{10\delta}{k^9} \geq 1-\frac{\delta}{k^5}$, in none of the (non-risky) chunks also we will lose a top-$k$ coin. Overall, 
this means that with probability 
\begin{align*}
1-\paren{\frac{2\delta}{k^3} + \frac{\delta}{k^5}} \geq 1-\frac{\delta}{2}
\end{align*}
we will not lose any top-$k$ coin throughout the stream, proving~\Cref{lem:topk-correct}.

\subsection*{Proof of \Cref{thm:top-k}}

We are ready to prove \Cref{thm:top-k}. The number of coin tosses for the algorithm immediately follows from conclusion of \Cref{lem:topk-sample}. Moreover, \Cref{lem:topk-correct} ensures that with probability at least $1-\frac{\delta}{2}$, all the top-$k$ coins
will be maintained in $\Kings \cup B$ by the end of the stream. Now we need one more simple lemma that states that the very final step of the algorithm also correctly returns the set of top-$k$ coins. The proof of this lemma follows from our earlier results (and also from known results in the literature since we can simply run any standard algorithm for finding top-$k$ coins on these set of $O(k)$ coins at the end). 

\begin{lemma}\label{lem:topk-buffer-sample}
With probability at least $1-\delta$, the algorithm will return the top-$k$ coins in line~\ref{linefinalsample} of algorithm $\algtopk$.
\end{lemma}
\begin{proof}
By~\Cref{lem:topk-correct}, with probability $1-\frac{\delta}{2}$, we have the top-$k$ coins in $\Kings \cup B$ by the end of the stream. Moreover, we have shown in the proof of \Cref{clm:risky-trial} that for any pair of a top-$k$ coin and a
 non-top-$k$ coin, if we toss both of them $s_{1}$ times, the probability for the latter to have a higher empirical bias than the former is at most $\frac{\delta^{2}}{k^{12}}$. This plus a union bound over the $10k$ coins implies that in this step
 also we may not return the top-$k$ coins with probability only $\frac{\delta}{2}$, finalizing the proof. 
\end{proof}

%% file: noisy-comp.tex

\newcommand{\element}{\ensuremath{\textnormal{\textsf{element}}}}

\section{Partition with Noisy Comparisons}\label{sec:comp}

In this section, we consider one applications of our techniques to the problem of top-$k$ recovery from noisy comparisons. 

\subsection*{Problem Definition}

In this problem, 
we have a collection of $n$ elements, denoted by $\set{\element_i}_{i=1}^{n}$, with an \emph{unknown total order} over these elements. The algorithm has a `noisy' access to this ordering: 
for any pairs of elements, the algorithm can query the order between the elements of this pair; with probability $1/2+\gamma$, the answer is according to the underlying total ordering, and with the remaining probability, the answer is arbitrary. 
The goal in the top-$k$ problem is to, given $\set{\element_i}_{i=1}^{n}$, parameters $k$ and $\gamma$, and query access to the underlying ordering, output the top largest $k$ elements according to this ordering, using a minimal number of 
queries. This problem is also sometimes referred to as the \emph{select} problem and its special case of $k=1$ is called the \emph{MAX} problem in the literature. 

We can model this problem in the streaming setting as before: the elements in $\set{\element_i}_{i=1}^{n}$ are arriving one by one in the stream and the algorithm is only allowed to store a limited number of these elements -- to query a pair of elements at any point, both elements are required to be in the memory of the algorithm. 

\subsection*{Our Results for the Top-$k$ Recovery Problem}

We obtain the following algorithms for this problem. 

\begin{theorem}\label{thm:comp}
	(Formalization of \Cref{res:noisy}) There exists streaming algorithms that given $n$ elements arriving in a stream, parameters $k$ and $\gamma$, and the confidence parameter $\delta$, with probability at least $1-\delta$, find
	 the top $k$ largest element in the underlying ordering in the noisy comparison model, using $O(k)$ memory and $O(\frac{n}{\gamma^2} \cdot \log{(k/\delta)})$ (noisy) comparisons. 
\end{theorem}

We shall note that the number of comparisons done by all our algorithms are \emph{optimal} (even in the absence of any memory restriction). \par
The algorithms can be directly obtained by showing that the top-$k$ recovery problem is mathematically equivalent to finding the $k$ most biased coin with gap at least $\gamma$. In this sense, one can directly apply our algorithms in~\Cref{thm:main} and~\Cref{thm:top-k}  (depending on whether $k\geq2$) to get the results. 

To show the mathematical equivalence of the two problems, the following lemma is crucial:
\begin{lemma}\label{lem:noisy-comp}
Let $\element_{1}$ and $\element_{2}$ be a pair of elements with true order $\element_{1} \succ \element_{2}$ (`$\succ$' here means `has a higher order than'). Suppose the noisy comparison will return a correct answer with probability $\frac{1}{2}+\gamma$, and we query the comparison $\frac{K}{\gamma^2}$ times and determine the element that wins the most times as the higher order element. Then, 
\begin{align*}
	\Pr\paren{\text{$\element_2$ is considered higher order by the algorithm}} \leq 2 \cdot \exp\paren{-\frac{1}{4} \cdot K}. 
\end{align*}
\end{lemma}
\begin{proof}
	The proof is similar to the proof of \Cref{lem:coin-comp}. Let us define two random variables:
	\begin{center}
	$\mathscr{C}_{r}$: The number of times the query returns $\element_{1}$ is greater than $\element_{2}$\\
	$\mathscr{C}_{w}$: The number of times the query returns $\element_{1}$ is smaller than $\element_{2}$
	\end{center}
	By definition and problem setup, we will have $\expect{\mathscr{C}_{r}}=(\frac{1}{2}+\gamma)\cdot\frac{K}{\gamma^{2}}$ and $\expect{\mathscr{C}_{w}}\leq \frac{1}{2}\cdot\frac{K}{\gamma^{2}}$. Applying \Cref{prop:chernoff} to both random variables will result in: 
	\begin{align*}
		\Pr\paren{\mathscr{C}_{r} \leq (\frac{1}{2}+\frac{\gamma}{2})\cdot\frac{K}{\gamma^{2}}} &\leq \exp\paren{-(\frac{\gamma}{2})^2 \cdot (\frac{K}{\gamma^2})} = \exp\paren{-\frac{1}{4} \cdot K}; \\
		\Pr\paren{\mathscr{C}_{w} \geq (\frac{1}{2}+\frac{\gamma}{2})\cdot\frac{K}{\gamma^{2}}} &\leq \exp\paren{-(\frac{\gamma}{2})^2 \cdot (\frac{K}{\gamma^2})} = \exp\paren{-\frac{1}{4} \cdot K}.
	\end{align*}
	A union bound on the events will conclude the proof. 
\end{proof}
By \Cref{lem:noisy-comp}, one can change the process of coin tossing and comparison in the algorithms in~\Cref{thm:main} and~\Cref{thm:top-k} to query and order with the same number of times. Thus, the properties of the algorithms we proved in \Cref{thm:main} and \Cref{thm:top-k} can directly lead to the proof of \Cref{thm:comp}.



%% file: eps-best.tex

\newcommand{\ptilde}{\widetilde{p}}
\newcommand{\pstar}{p^{*}}
\newcommand{\msubrout}{\ensuremath{\textnormal{\textsc{MAB Subroutine}}}\xspace}
\newcommand{\infsubrout}{\ensuremath{\textnormal{\textsc{MAB-INF Process}}}\xspace}
\newcommand{\BK}{\ensuremath{\mathcal{B}(K)}}

\section{Approximate Exploration in Stochastic Multi-Armed Bandits}\label{sec:eps-best}


We consider another application of our algorithms, this time to the approximate exploration problem in stochastic multi-armed bandits under the $\PAC(\eps, \delta)$ framework. 

\subsection*{Problem Definition}
In the stochastic multi-armed bandit (MAB) problem, we have a collection of $n$ arms $\set{\arm_i}_{i=1}^{n}$. Each sample (or pull) of any $\arm_i$ 
results in a reward in $[0,1]$ sampled from an unknown distribution with mean $\mu_i \in [0,1]$\footnote{Our results extend verbatim to any Sub-Gaussian reward distribution with no assumption on range of the rewards. This, to the best our knowledge,
is the common characteristic of all prior work on exploration in MAB as well, and simply follows from the fact that the Chernoff-Hoeffding inequality used in the proofs extends directly to these distributions. As such, we omit the details.}. 
For a parameter $\eps \in (0,1)$, we say that an $\arm_i$ is an \emph{$\eps$-best arm} if its expected reward is at most $\eps$ smaller than the expected reward of the maximum (the best arm), or alternatively $\mu_i \geq \max_j \mu_j - \eps$. 
In the exploration problem, our goal is to, given the arms $\set{arm_i}_{i=1}^{n}$ and a parameter $\eps > 0$, return any $\eps$-best arm using a minimal number of arm pulls with probability $1-\delta$ -- a task that fits naturally in the classical $\PAC(\eps, \delta)$ framework.

We study this problem in the streaming model as follows: The arms are arriving one by one in a stream and the algorithm needs to store each arm explicitly if it wants to 
pull it at some later point in the stream as well. Based on the above model, our results on the $\eps$-best arm problem include the following algorithms
\begin{itemize}
\item An algorithm using a memory of a single arm and $O(\frac{n}{\eps^2} \cdot \log{(1/\delta)})$ arm pulls assuming at least $\eps$ gap between the largest (expected) reward and the second largest reward.
\item An algorithm using a memory of $O(\logstar{(n)})$ arms and $O(\frac{n}{\eps^2} \cdot \log{(1/\delta)})$ arm pulls. 
\item An algorithm using a memory of two arms and $O(\frac{n}{\eps^2} \cdot \log{(1/\delta)}+\log^{2}(n)\cdot \frac{\log^{2}\paren{1/\delta}}{\eps^{3}})$ arm pulls. 
\end{itemize}
(The latter two algorithms do not require making an assumption on the gap between the largest and second largest reward.)

When a gap of at least $\eps$ exists between the best and second-best arms,  the problem can simply be solved by our main algorithm $\mainalg$ -- nothing needs to be changed except the notation. Hence, we omit the repetition of excessive technical details and focus on the second and the third algorithms in this section.


\subsection{The $O(\logstar(n))$ Memory Algorithm}

We design the following streaming algorithm. 

\begin{theorem}\label{thm:eps-best-logstar}
There exist a streaming algorithm that given $n$ arms arriving in a stream, the approximation parameter $\eps \in (0,1)$, and the confidence parameter $\delta$, with probability at least $1-\delta$, finds an $\eps$-best arm using a memory of $O(\logstar{(n)})$ arms and $O(\frac{n}{\eps^2} \cdot \log{(1/\delta)})$ arm pulls. 
\end{theorem}
For the problem of finding the $\eps$-best arm without the gap guarantee, our $\mainalg$ algorithm does not work in general. The issue here is that if a bunch of arms with gaps far smaller than $\eps$ arrive in a consecutive manner, the less stronger arms will have non-trivial probabilities to replace the stronger ones. And if this type of event happens over time, arms with gap larger than $\eps$ will eventually become selected as the king in the algorithm. \par

We tackle the above problem by iteratively refining the gap of selecting arms. Specifically, we leverage the framework of the $\log^{*}(n)$ space algorithm in \Cref{sec:prelim-algs}, and repetitively narrowing the gap of $\eps_{l}=O(\frac{\eps}{2^{l-1}})$ at each layer $l$. Since the number of arms with the $\log^{*}(n)$ space algorithm will rule out arms by a tower factor, we will have enough additional budget to pay for the up-sampling factor. \par

Another concern is that with the $\logstar(n)$-memory framework, the numbers of challenges in the higher levels become large, and an arm with more than $\eps$ gap may eventually be selected. We address this issue by making the challengers' requirement more demanding and it can only become the best arm of the level if it has a higher empirical reward than the \emph{best record} of the stored arm, as opposed to the repeatedly comparing this arm with the best arm on this level\footnote{This step is what
was missing from our original algorithm for this problem in the previous version of the paper which led to the aforementioned error discussed in the Introduction.}. The algorithm can be shown as follows:

\begin{tbox}
	\textbf{Algorithm}: 
	
	\medskip
	
	Parameters: 
	\begin{align*}
	& \set{\eps}_{\ell\geq 1}: \eps_{\ell}=\frac{\eps}{10\cdot2^{\ell-1}} \tag{gap parameter at each level}\\
	&\set{r_\ell}_{\ell\geq 1}: r_{1}:=4,\quad r_{\ell+1} = 2^{r_{\ell}}; \qquad \set{\beta_{\ell}}_{\ell\geq 1}: \beta_{\ell}=\frac{1}{\eps_{\ell}^{2}}; \quad\beta = \frac{1}{\eps^{2}}\tag{intermediate variables to define $s_\ell$ and $c_\ell$}\\
	&\set{s_{\ell}}_{\ell\geq 1}: s_{\ell} = 4\beta_{\ell}(\ln(\frac{1}{\delta})+3r_{\ell}) \tag{number of samples per each level}\\
	& \set{c_{\ell}}_{\ell\geq 1} c_{1} = 2^{r_{1}},\quad  c_{\ell}=\frac{2^{r_{\ell}}}{2^{\ell-1}} (\ell \geq 2) \tag{the bound for restarting the counter of each level}.
	\end{align*}
	
	Counters: $C_{1}, C_{2}, ..., C_{t} \qquad t=\left \lceil{\log^{*}(n)}\right \rceil+1$;\\
	Reward records: $p_{1}^{*}$, $p_{2}^{*}$, ..., $p_{t}^{*}$, initialize with $0$; \\
	Stored arms: $\arm^{*}_{1}, \arm^{*}_{2}, ..., \arm^{*}_{t}$ the most bias coin of $\ell$-th level.
	
	\medskip
	
	\begin{itemize}
	\item For each arriving $\arm_i$ in the stream do:
	\begin{enumerate}[label=($\arabic*$)]
		\item Read $\arm_{i}$ to memory.
		\item \textbf{\underline{Aggressive Selective Promotion}:} Starting from level $\ell=1$: 
		\begin{enumerate}
			\item\label{linesample}  Sample $\arm_{i}$ for $s_{\ell}$ times and get $\phat_{\arm_{i}}$. Drop $\arm_{i}$ if $\phat_{\arm_{i}}<p_{\ell}^{*}$;
			\item Otherwise, replace $\arm^{*}_{\ell}$ with $\arm_{i}$ and set $p_{\ell}^{*}=\phat_{\arm_{i}}$.
			\item Increase $C_{\ell}$ by 1.
			\item If $C_{\ell}=c_{\ell}$, send $\arm^{*}_{\ell}$ to the next level by calling Line~\ref{linesample} with $(\ell = \ell+1)$.
		\end{enumerate}
	\item Return $\arm_{t}^{*}$ as the selected most bias coin.
	\end{enumerate}
	\end{itemize}
\end{tbox}
\noindent
At a first glance, the algorithm is very similar to the $\log^{*}(n)$ space algorithm for the coin tossing problem -- in addition to the change of notation, the only differences here are that we add a $\frac{1}{10\cdot 2^{\ell-1}}$ factor for each level of $\eps$, and the empirical reward of each arm will be only tested once. We will show that, after adding this up-sampling factor, the overall sample complexity will still be $O(\frac{n}{\eps^{2}}\,\log(\frac{1}{\delta}))$(consistent with the similar result of the $O(\logstar(n))$-memory coin tossing algorithm in \Cref{sec:prelim-algs}); and the second modification guarantees the correctness of the algorithm. Formally, we claim:

\begin{lemma}\label{lem:alg-epsbest-sample}
The sample complexity of the algorithm is $O(n\cdot\beta\cdot\log(\frac{1}{\delta}))=O(
\frac{n}{\eps^{2}}\cdot\log(1/\delta))$.
\end{lemma}
\begin{proof}
Most part of this proof can be directly taken from the proof of \Cref{lem:alg3-sample}. Recall that at each level $\ell$, the number of arms to be processed will be bounded by  $\frac{n}{\prod_{i=1}^{\ell-1}c_{i}}$. Thus, the total number of sampling at level $\ell$ is $\frac{4\,n\cdot \beta_{\ell}\cdot(\ln(\frac{1}{\delta})+3r_{\ell})}{\prod_{i=1}^{\ell-1}c_{i}}$ times. Observe that by the definition, there is $\beta_{\ell}=\beta\cdot (10\cdot 2^{\ell-1})^{2}$; hence, the total number of sampling among all levels should be:
\begin{align*}
\sum_{\ell=1}^{\left \lceil{\log^{*}(n)} \right \rceil+1} \frac{4n\cdot\beta_{\ell}\cdot(\ln(\frac{1}{\delta})+3r_{\ell})}{\prod_{i=1}^{\ell-1}c_{i}}&= 4n\cdot\beta\sum_{\ell=1}^{\left \lceil{\log^{*}(n)} \right \rceil+1} \frac{(\ln(\frac{1}{\delta})+3r_{\ell})}{\prod_{i=1}^{\ell-1}c_{i}}\cdot (10\cdot2^{\ell-1})^{2}\\
&\leq 400n\cdot\beta\sum_{\ell=1}^{\infty}(\frac{\ln(\frac{1}{\delta})}{c_{\ell-1}}+ \frac{3r_{\ell}}{c_{\ell-2}c_{\ell-1}})\cdot 2^{2\ell-2}\\
&= 400n\cdot\beta(\sum_{\ell=1}^{\infty} \frac{\ln(\frac{1}{\delta})\cdot4^{\ell-1}}{c_{\ell-1}}+\sum_{\ell=1}^{\infty} \frac{3 \cdot2^{3\ell-4}}{c_{\ell-2}})\\
&\leq 400n\cdot\beta\cdot2\ln(\frac{1}{\delta}) + 1200n\beta(O(1)+\sum_{\ell=5}^{\infty} \frac{2^{4\ell-7}}{r_{\ell-2}}) \tag{$\sum_{\ell=1}^{4}r^{\ell}2^{2\ell-2} = O(1)$}\\
&\leq 800n\cdot\beta\ln(\frac{1}{\delta}) + O(1)\cdot n\beta. \tag{$\sum_{\ell=5}^{\infty} \frac{2^{4\ell-7}}{r_{\ell-2}}<1$}
\end{align*}
And in asymptotic notation, this is $O(n\cdot\beta\cdot\log(\frac{1}{\delta}))=O(
\frac{n}{\eps^{2}}\cdot\log(1/\delta))$.
\end{proof}
\begin{lemma}\label{lem:alg-epsbest-correct}
With probability at least $1-\delta$, the coin selected by the algorithm is an $\eps$-best arm.
\end{lemma}
\subsection*{Proof of \Cref{lem:alg-epsbest-correct}}
This proof is similar to the proof of \Cref{lem:alg3-correct}. The major difference is that instead of claiming the consistent selection of the best arm, we claim here that after each level $\ell$, the selected best arm $\arm^{*}_{\ell}$ has at most $\eps_{\ell}$ gap with the best arm of the previous layer $(\ell-1)$. Specifically, according to \Cref{lem:coin-comp}, at each level $\ell$, for an arm with reward $p$, the probability for its empirical reward $\phat$ to deviate beyond $\frac{\eps_{\ell}}{2}$ is small:
\begin{equation}
\label{equ:arm-reward-concentration}
\begin{split}
&\Pr(\phat<p-\frac{\eps_{\ell}}{2}) \leq \exp(-\frac{1}{4}\cdot 4(\ln(\frac{1}{\delta})+3r_{\ell})) \leq \frac{\delta}{2^{3r_{\ell}}}\\
&\Pr(\phat>p+\frac{\eps_{\ell}}{2}) \leq \exp(-\frac{1}{4}\cdot 4(\ln(\frac{1}{\delta})+3r_{\ell})) \leq \frac{\delta}{2^{3r_{\ell}}}
\end{split}
\end{equation}
Consequently, let $p_{\ell}$ be the highest reward in one iteration of level $\ell$ (counter from $0$ to $c_{\ell}$), the following claims hold.

\begin{claim}\label{clm:p-ell-large}
Fix a level $\ell$, with probability at least $(1-\frac{\delta}{2^{3r_{\ell}}})$, $p^{*}_{\ell}\geq p_{\ell} - \frac{\eps_{\ell}}{2}$.
\end{claim}
\begin{proof}
This is a natural corollary of the first inequality of \Cref{equ:arm-reward-concentration}, which says that when we pull the arm with reward $p_{\ell}$ for $s_{\ell}$ times, the probability for the empirical reward to be less than $p_{\ell} - \frac{\eps_{\ell}}{2}$ is at most $\frac{\delta}{2^{3r_{\ell}}}$. Hence, when we compare this $\phat$ with $p^{*}_{\ell}$, either $p^{*}_{\ell}$ is already at least $p_{\ell} - \frac{\eps_{\ell}}{2}$, or we replace $p^{*}_{\ell}$ with the larger value. 
\end{proof}

\begin{claim}\label{clm:p-other-small}
Fix a level $\ell$, with probability at least $(1-\frac{\delta}{2^{3\ell-1}})$, any \arm with $p<p_{\ell}-\eps_{\ell}$ has empirical reward $\phat < p_{\ell} - \frac{\eps_{\ell}}{2}$.
\end{claim}
\begin{proof}
This is a natural corollary of the second inequality of \Cref{equ:arm-reward-concentration}. Based on the inequality, conditioning on an arm is with reward less than $p_{\ell}-\eps_{\ell}$, we have $\Pr(\phat\geq p_{\ell}-\frac{\eps_{\ell}}{2})\leq \Pr(\phat>p+\frac{\eps_{\ell}}{2})  \leq \frac{\delta}{2^{3r_{\ell}}}$. Applying a union bound over at most $c_{\ell}=\frac{2^{r_{\ell}}}{2^{\ell-1}}$ arms, the probability for `reward overflow' to happen on any arm at level $\ell$ is at most $\frac{\delta}{2^{3r_{\ell}}} \cdot \frac{2^{r_{\ell}}}{2^{\ell-1}} \leq \frac{\delta}{2^{3\ell-1}}$. 
\end{proof}
\noindent
Combining \Cref{clm:p-ell-large} and \Cref{clm:p-other-small}, we can prove follows:
\begin{claim}\label{clm:arm-reward-gap}
After the $\arm$ with the highest reward (the \emph{best} arm) joins the stream, with probability at least $1-\delta$, at any level $\ell$, there is at least one arm with at most $\sum_{i=1}^{\ell}\eps_{i}$ reward gap between the best arm.
\end{claim}

\begin{proof}
We can apply a union bound over the events of \Cref{clm:p-ell-large} and \Cref{clm:p-other-small} and get that with probability at least $(1-\frac{\delta}{2^{3r_{\ell}}}-\frac{\delta}{2^{3\ell-1}}) \geq (1-\frac{\delta}{2^{2\ell}})$, there are $(p^{*}_{\ell}\geq p_{\ell} - \frac{\eps_{\ell}}{2})$ \emph{and} any \arm with $p<p_{\ell}-\eps_{\ell}$ has empirical reward $\phat < p_{\ell} - \frac{\eps_{\ell}}{2}$. Hence, we conclude that with probability at least $(1-\frac{\delta}{2^{2\ell}})$, no arms with reward more than $\eps_{\ell}$-lower from the arm with the highest reward of the current iteration will be sent to the higher level. Thus, by applying a union bound over all the levels, we can bound the probability of selecting any arm at level $\ell$ with reward gap $\geq \eps_{\ell}$ from the best arm selected on level $(\ell-1)$ as
\begin{equation*}
\begin{split}
\sum_{\ell=1}^{\left \lceil{\log^{*}(n)} \right \rceil+1} \frac{\delta}{2^{2\ell}} & \leq \delta\sum_{\ell=1}^{\infty} \frac{1}{2^{2\ell}}
< \delta
\end{split}
\end{equation*}
During the stream of arms, the best arm $\arm^{*}$ will eventually join at the first level. Then, with probability at least $(1-\delta)$, since any reward gap between two levels will not be greater than $\eps_{\ell}$, there should be at least one arm at any level with  gap $\leq \sum_{i=1}^{\ell}\eps_{i}$ from the best arm.
\end{proof}

Conditioning on the event of \Cref{clm:arm-reward-gap}, the gap between the best arm and the $\arm^{*}_{\ell}$ at any level $\ell$ is bounded. Accumulating the gap among every level and summing up will give us
\begin{align*}
\sum_{\ell=1}^{\left \lceil{\log^{*}(n)} \right \rceil+1} \eps_{\ell} & = \sum_{\ell=1}^{\left \lceil{\log^{*}(n)} \right \rceil+1} \frac{\eps}{10\cdot 2^{\ell-1}}
\leq \frac{\eps}{30}\sum_{\ell=1}^{\infty} \frac{1}{2^{\ell-1}}
\leq\frac{\eps}{3} < \eps.
\end{align*}
That is to say, the cumulative gap between the best arm and the selected arm is less than $\eps$, which satisfies the definition of selecting an $\eps$-best arm. Therefore, we verify the correctness of the algorithm; the space complexity is clearly $O(\logstar)$, which proves \Cref{thm:eps-best-logstar}.

\subsection{The $O(1)$ Memory Algorithm}

We now introduce an algorithm that uses only 2 arms memory and is more closely related to the ideas behind $\mainalg$ compared to our algorithm in the previous section. 
\begin{theorem}\label{thm:eps-best-2-arm}
There exist a streaming algorithm that given $n$ arms arriving in a stream, the approximation parameter $\eps \in (0,1)$, and the confidence parameter $\delta$, with probability at least $1-\delta$, finds an $\eps$-best arm using a memory of two arms and $O(\frac{n}{\eps^2} \cdot \log{(1/\delta)}+\log^{2}(n)\cdot \frac{\log^{2}\paren{1/\delta}}{\eps^{3}})$ arm pulls. 
\end{theorem}
\noindent
We remark that when $\eps$ is not too small ($\eps > \omega(\frac{\log^{2}(n)}{n})$), our algorithm achieves $O(\frac{n}{\eps^2} \cdot \log{(1/\delta)})$ sample complexity and a memory of two arms, which almost matches our result in \Cref{thm:main} with known gap guarantee.

\subsubsection*{High-level Overview of the algorithm}
As we have discussed before, the central issue that blocks our $\mainalg$ algorithm to be applied to the $\eps$-best arm is the lack of a gap parameter. Per the perspective of the algorithm, even regardless of the  correctness, such a parameter is essential. Therefore, a naive idea is to \emph{artificially} set a gap parameter $\Delta$ (say, $\Delta=\frac{\eps}{3}$), and run the $\mainalg$ algorithm with such a artificial gap. Denoting the reward of the stored $\king$ as $\ptilde$, we consider the following \emph{dreamland} scenario: when the $\king$ is an $\eps$-best arm, all the arriving arms are with reward less than $\ptilde - \Delta$. In this case, we can simply run the $\mainalg$ algorithm with the artificially specified $\Delta$ and discard the $\king$ whenever its budget is exhausted, and the correctness and the sample complexity almost immediately follows from \Cref{thm:main}. 

Our challenge is to address the situations outside of the dreamland. More specifically, when arms with rewards of more than $\ptilde-\Delta$ arrive, it is unclear whether there is a unified way to treat them: if we simply replace the $\king$ if the budget is exhausted, the stream of arms may eventually replace the $\king$ with an arm of reward gap more than $\eps$; on the other hand, if we allow extra arm pulls, it remedies the case when the $\king$ is a $\eps$-best $\arm$, but the number of arm pulls blows up when the number of arriving arms with rewards of more than $\ptilde-\Delta$ is large. 

Our solution to the above dilemma is to separate the cases for the $\king$ as the best and other arms when the rewards of the arriving arms are more than $\ptilde-\Delta$. Intuitively, if the current $\king$ is using a lot of arm pulls, it is either because it is not an $\eps$-best arm, or because there are a lot of arms with reward at least $\ptilde-\Delta$. However, if it is the latter case, we can uniformly at random sample $O(\frac{\log(n)}{\Delta}) = O(\frac{\log(n)}{\eps})$ arms, and with a high probability, one of the $\eps$-best arms will be sampled. In this way, we can make it safe to discard a $\king$ and update the estimation whenever it costs too many arm pulls. 

Another issue is how to get a reliable estimation of $\ptilde$. The naive way to get good estimations for \emph{every} arm with high probability is to pull each arm around $O(\log(n))$ times, but this blows up the sample complexity. Therefore, we use an alternative idea that maintains the estimation $\ptilde$ as $t\cdot \frac{\eps}{30}$, where $t$ is an integer that is updated by the algorithm (think of $t$ as a `level' -- although different from the notion of level we used in the challenge). In this way, there are at most $O(\frac{1}{\eps})$ levels of estimations of the rewards. By carefully designing the algorithm, we can show that the value of $\ptilde$ is only different from the real reward of the $\king$ up to a fraction of $\eps$ when the best arm becomes the $\king$. Therefore, the above intuitions can be applied to get the $\eps$-best arm.

To formalize the above intuitions, we introduce the notion of \emph{Top} and \emph{Bottom} arms as follows. Note that these definitions are self-contained and more general than the context of the $\eps$-best arm. For any reward $\ptilde$ and any gap $\Delta$, we can define the top and bottom arms.

\begin{definition}[Top and Bottom Arms]
\label{def:toparm}
Let $\ptilde$ be a given reward, $\Delta$ be the gap parameter, and $p_{i}$ be the reward of $\arm_{i}$. We say the set of $\arm$s $\mathcal{T}_{\ptilde, \Delta}:=\{\arm_{i}|(\ptilde - \frac{2}{3}\Delta)\leq p_{i} \leq (\ptilde + \frac{1}{3}\Delta)\}$ is the set of \emph{($\ptilde$, $\Delta$)-Top} $\arm$s; Similarly, we say the set of $\arm$s $\mathcal{B}_{\ptilde, \Delta}:=\{\arm_{i}|p_{i}<(\ptilde -  \frac{2}{3}\Delta)\}$ is the set of \emph{($\ptilde$, $\Delta$)-Bottom} $\arm$s.
\end{definition}

\subsubsection*{Introduction to $\msubrout$}
We now formalize the above intuitions. To help understand the properties of the algorithm, we first present a modified $\msubrout$ for multi-armed bandits. Compared to the original subroutine used in \Cref{thm:main}, the new challenge subroutine does \emph{not} allow the challenger to always bringing up the level of challenge. If at a certain level, the number of arm pulls has become $O(\frac{\log(n)}{\Delta^{2}})$, we explicitly terminate the process and determine the winner. This method could fail to preserve the $\king$ if the the reward of the $\king$ is only slightly higher than the challenger; however, an arm whose reward is very close to the best suffices the purpose of the $\eps$-best arm problem (we will see more details in the analysis). On the other hand, by capping the sample complexity to $O(\frac{\log(n)}{\Delta^{2}})$, the sample complexity for each given $\ptilde$ can be bounded by $O(\frac{n}{\Delta})$ as long as the ($\ptilde$, $\Delta$)-Top arms is of size at most $O(\frac{\Delta\cdot n}{\log(n)})$ -- the importance of this will be clear in the final algorithm. 

The details of $\msubrout$ is as follows. 

\begin{tbox}
	\textbf{Algorithm} $\msubrout$: 
	
	\medskip
	
	Parameters: 
	\begin{align*}
	& L=\log\log(n) \tag{number of levels}\\
	& \set{\alpha_\ell}_{\ell=1}^{L}: \alpha_\ell = \min(3^{\ell}, \log(n)); \tag{intermediate variable to define $s_\ell$}\\
	& \set{s_\ell}_{\ell=1}^{L}:  s_\ell :=  \frac{180}{\Delta^2} \cdot \log{(1/\delta)} \cdot \alpha_\ell; \tag{number of arm pulls at each level}\\
	&b := \frac{180}{\Delta^2} \cdot C \cdot \log{(1/\delta)} + s_1. \tag{$C$ is a sufficiently large constant}
	\end{align*}
	
	\medskip
	\begin{enumerate}
	    \item Let $\arm_{1}$ be the \emph{incumbent} $\arm$ and $\arm_{2}$ be the \emph{challenger} $\arm$.
		\item For level $\ell = 1$ to $L$ do:
		\begin{enumerate}
		    \item If $\alpha_{\ell}=\log(n)$, declare this level as the final.
			\item Pull both arms $s_{\ell}$ times, and get empirical rewards $\phat_{1}$ and $\phat_{2}$. 
			\item If $\phat_{1}>\phat_{2}$, drop $\arm_{2}$ and return.
			\item Otherwise, let $\ell = \ell+1$ and go to the next level. 
		\end{enumerate}
		\item If $\phat_{2}\geq\phat_{1}$ at every level $\ell$, then we declare $\arm_{2}$ defeats $\arm_{1}$ and let $\arm_{2}$ become the new \emph{incumbent} $\arm$.
	\end{enumerate}
\end{tbox}

\noindent
And we present claims characterizing the property of $\msubrout$. We start with bounding the sample complexity.
\begin{claim}
\label{clm:arm-sample-limite}
The number of possible arm pulls to be used by a single arm is at most $540 \cdot\frac{\log(n)}{\Delta^2} \cdot  \log{(1/\delta)}$.
\end{claim}
\begin{proof}
We note that the number of arm pulls on the highest level is $180\cdot \frac{\log(n)}{\Delta^2} \cdot  \log{(1/\delta)}$, the second-last level has number of arm pulls at most $180 \cdot\frac{\log(n)}{\Delta^2} \cdot  \log{(1/\delta)}$, and the rest of the levels decreases by a factor of $\frac{1}{3}$. Therefore, the number of total possible arm pulls for a single arm at most
\begin{align*}
180 \cdot\frac{\log(n)}{\Delta^2} \cdot  \log{(1/\delta)}\cdot\paren{1 + \sum_{\ell=L-1}^{1} (\frac{1}{3})^{L-\ell-1}} &\leq 180 \cdot\frac{\log(n)}{\Delta^2} \cdot  \log{(1/\delta)}\cdot\paren{1 + \sum_{\ell=0}^{\infty} (\frac{1}{3})^{\ell-1}} \\
& = \frac{5}{2} \cdot 180 \cdot\frac{\log(n)}{\Delta^2} \cdot  \log{(1/\delta)}\cdot\paren{1 + \sum_{\ell=0}^{\infty} (\frac{1}{3})^{\ell-1}} \tag{$\sum_{\ell=0}^{\infty} (\frac{1}{3})^{\ell-1} = \frac{3}{2}$}\\
&\leq 540 \cdot\frac{\log(n)}{\Delta^2} \cdot  \log{(1/\delta)}.
\end{align*}
\end{proof}

We now turns to the empirical rewards for the incumbent and challenger arms. We begin with stating the completeness of the incumbent arm: it does \emph{not} get defeated if their reward gap is sufficiently large.
\begin{claim}
\label{clm:challenge-complete}
If $p_{2}<(p_{1}-\frac{\Delta}{3})$, then $\phat_{1}>\phat_{2}$ with probability at least $1-\frac{\delta}{4n}$. 
\end{claim}
\noindent
\begin{proof}
Note that for $\arm_{2}$ to replace $\arm_{1}$, it has to have a higher reward than $\arm_{1}$ at \emph{every} level. And at the top level ($\ell = \log\log(n)$ or when $\alpha\geq\log(n)$), we have $s_{\ell}=\frac{20}{(\Delta/3)^{2}}\cdot\log(1/\delta)\cdot \log(n)$ arm pulls. Therefore, we have
\begin{align*}
\Pr\paren{\text{$\phat_{2}>\phat_{1}$ on every level}} & \leq \Pr\paren{\text{$\phat_{2}>\phat_{1}$ on the top level}}\\
&\leq 2\cdot \exp(-\frac{1}{4}\cdot 20 \cdot\log(1/\delta)\cdot \log(n)) \tag{By \Cref{lem:coin-comp}}\\
& = 6 \cdot 2^{-5} \cdot \delta \cdot \exp(-\log(n)) \\
& \leq \frac{\delta}{4n}.
\end{align*}
\end{proof}

We then show the soundness of $\msubrout$: if the challenger arm is indeed better, the incumbent arm will be defeated.
\begin{claim}
\label{clm:challenge-soundness}
If $p_{2}>(p_{1}+\Delta)$, then $\arm_{2}$ eventually becomes the $\king$ with probability at least $(1-\frac{\delta}{8})$. 
\end{claim}
\noindent
\begin{proof}
We can compute the probability for $\arm_{2}$ to have a smaller empirical reward than $\arm_{1}$ at \emph{any} level is at most $\frac{\delta}{8}$. More formally,
\begin{align*}
\Pr\paren{\text{$\phat_{2}<\phat_{1}$ on any level}} & \leq \sum_{\ell=1}^{L}\Pr\paren{\text{$\phat_{2}<\phat_{1}$ on level $\ell$}}\\
&\leq \sum_{3^{\ell}\leq \log(n)}\Pr\paren{\text{$\phat_{2}<\phat_{1}$ on level $\ell$}} + \sum_{3^{\ell}>\log(n)}^{L}\Pr\paren{\text{$\phat_{2}<\phat_{1}$ on level $\ell$}}\\
&\leq \sum_{3^{\ell}\leq \log(n)} 2\cdot \exp(-45 \cdot\log(1/\delta)\cdot 3^{\ell}) + \sum_{3^{\ell}>\log(n)}^{L} \frac{\delta}{16n} \tag{By \Cref{lem:coin-comp} and changing $\frac{\Delta}{3}$ to $\Delta$ in \Cref{clm:challenge-complete}}\\
& \leq \sum_{\ell=1}^{\infty} 2\cdot \exp(-45 \cdot\log(1/\delta)\cdot 3^{\ell}) + \frac{\delta}{16} \tag{$L\leq \log\log(n)<< n$} \\
&\leq 2\cdot \frac{\delta}{16} \cdot \sum_{\ell=1}^{\infty} \exp(-3^{\ell}) +\frac{\delta}{16}\\
& \leq \frac{\delta}{8} \tag{$\sum_{\ell=1}^{\infty} \exp(-3^{\ell})<\frac{1}{10}$}.
\end{align*}
\end{proof}

Finally, we show that with a good probability, an incumbent arm with reward $p_1$ will never be defeated by challengers from ($p_1$, $\Delta$)-Bottom arms $\mathcal{B}_{p_1, \Delta}$ (see \Cref{def:toparm}) in $\msubrout$:
\begin{claim}
\label{clm:challenge-arm-sample}
Let an $\arm$ with reward $p_1$ be the current incumbent arm, let $\arm_{\BK}$ be the $K$-th arriving arm from $\mathcal{B}_{p_1, \Delta}$ and $X_{\BK}$ to be the number of arm pulls used by the arms in $\mathcal{B}_{p_1, \Delta}$ to challenge the stored arm. With probability at least $1-\frac{\delta}{8}$, we have $X_{\BK}<K\cdot b$.
\end{claim}
\noindent
\begin{proof}
The claim is a natural corollary of \Cref{lem:main-keep}. In the interest of conciseness, we omit repeating the technical details, and instead discuss how can we use the result of \Cref{lem:main-keep} to prove \Cref{clm:challenge-arm-sample}. To this end, we define $\infsubrout$ as a process that follows $\msubrout$ with $L=+\infty$, i.e. the number of levels is allowed to go to infinity. Define $Y_{\BK}$ as the number of arm pulls used by $\infsubrout$ by the arms in $\mathcal{B}_{p_1, \Delta}$; conditioning on the process has not terminated and both $\infsubrout$ and $\msubrout$ continue to the $\mathcal{B}(K)$-th arm, we (deterministically) have $X_{\BK}\leq Y_{\BK}$. As such, an upper on $Y_{\BK}$ is sufficient to prove our desired statement.

Note that for the arriving arms in $\mathcal{B}_{p_1, \Delta}$, they necessarily have rewards at most $(p_{1}-\frac{2}{3}\Delta)$ by definition. The number of arm pulls at each level is
\[s_{\ell}=\frac{180}{\Delta^{2}}\cdot \log(1/\delta)\cdot 3^{\ell} = \frac{80}{(\frac{2}{3}\cdot\Delta)^{2}}\cdot \log(1/\delta)\cdot 3^{\ell},\]
and the budget is given as $\frac{80}{(\frac{2}{3}\cdot\Delta)^{2}}\cdot C\cdot \log(1/\delta)$ for a sufficiently large constant $C$. By \emph{restricting} to the arms in $\mathcal{B}_{p_1, \Delta}$, we can reduce this process to \emph{keeping} the best arm with gap parameter $\frac{2}{3}\cdot \Delta$. As such, by using \Cref{lem:main-keep} with $\Delta'=\frac{2}{3}\cdot \Delta$ and $\delta'=\frac{\delta}{4}$, we conclude that with probability at least $\frac{\delta}{8}$, there is $Y_{\BK}<K\cdot b$ for all $K$.
\end{proof}

\subsubsection*{The 2-arm memory algorithm}
We are now ready to introduce the algorithm based on the above $\msubrout$. We assume $\pstar>\eps$ as otherwise one can arbitrarily pick an arm to find the $\eps$-best arm. The algorithm is shown as follows.

\begin{tbox}
	\textbf{Algorithm}:\\
	\medskip
	Parameters:\\
	Estimation levels: $t \in \{1,2,\cdots,\frac{30}{\eps}\}$.\\
	Counter of arm pulls at each level: $\set{A_{t}}_{t=1}^{\frac{30}{\eps}}: A_{t}=0$.\\
	Budget of at each level: $\set{B_{t}}_{t=1}^{\frac{30}{\eps}}$\\
	Budget of the top arms: $B=720 \cdot 100 \cdot n \cdot \frac{\log(1/\delta)}{\eps}$.
    \medskip
	\begin{enumerate}[label=\alph*)]
	    \item Run the following procedures in parallel:
	    \item \textbf{\underline{Full Stream Procedure}:}
	    \begin{enumerate}[label=\arabic*.]
	        \item Initialize all $B_{t} = B$.
	        \item Initialize the best reward estimation $\ptilde=\eps$.
    	        \item With estimation level $t\in\{1,2,\cdots,\frac{30}{\eps}\}$, while $\ptilde\leq 1$ do:
	            \begin{enumerate}
	            \item \label{line:increasebudget} For each incoming arm, increase the budget $B_{t}$ by $b= \frac{180}{\Delta^2} \cdot C \cdot \log{(1/\delta)} + s_1$ as in $\msubrout$.
        		\item Run $\msubrout$ with $p_{1}=\ptilde-\frac{4}{30}\eps$ and $\Delta=\frac{\eps}{10}$. Increase $A_{t}$ by the number of arm pulls used.
        		\item In the case that the $\king$ is defeated, pull the arriving arm ($\arm_{2}$) for $180\cdot \frac{1}{(\eps/30)^{2}}\cdot\log(1/\delta)$ times, and record the empirical reward of $\phat_{2}$. 
        		\item \label{line:increasetbytarget} If $\arm_{2}$ has an empirical reward $\phat_{2}\geq \ptilde +\frac{1}{30}\eps$, then update $t$ such that $(t-1)\cdot \frac{\eps}{30} < \phat_{2}\leq t\cdot \frac{\eps}{30}$ and start with the new level.
        		\item \label{line:increasetbyone} Else, if $A_{t}>B_{t}$, update $t=t+1$ and start with the new level.
        		\end{enumerate}
        		\item If $t>\frac{30}{\eps}$, terminate the process and return \textsc{FAIL}.
		\end{enumerate}
		\item \textbf{\underline{Fraction Stream Procedure}:}
		\begin{enumerate}[label=\arabic*.]
    		\item Sample $4\cdot \log(n)\cdot \frac{\log(1/\delta)}{\eps}$ arms from the stream.
    		\item Pull each of the arms $ 100\cdot\frac{180}{\eps^{2}}\cdot\log(1/\delta)\cdot \log(n)$ times, return the one with the maximum reward.
		\end{enumerate}
		\item If \underline{Full Stream Procedure} returns an arm, then pick it as the selected $\eps$-best arm; otherwise, pick the output of \underline{Fraction Stream Procedure} as the $\eps$-best arm.
	\end{enumerate}
\end{tbox}
\noindent
\subsubsection*{The analysis of the algorithm}
We first show the space complexity and the sample complexity of the algorithm, which are relatively straightforward to verify.
\begin{lemma}[Space Complexity]
\label{lem:2-arm-mab-memory}
The algorithm uses a memory of two arms.
\end{lemma}
\begin{proof}
The memory for the \underline{Full Stream Procedure} is one arm, and the sampling of the arms in the \underline{Fraction Stream Procedure} can be done in-place with just one arm using reservoir sampling and maintaining the maximum reward arm. 
\end{proof}
\begin{lemma}[Sample Complexity]
\label{lem:2-arm-mab-sample}
The sample complexity of the algorithm is 
\[
O\paren{n\cdot \frac{\log(1/\delta)}{\eps^{2}} + \log^{2}(n)\cdot\frac{\log^{2}(1/\delta)}{\eps^{3}}}.
\]

\end{lemma}
\begin{proof}
We first show for the \underline{Full Stream Procedure} the sample complexity is $O(n\cdot \frac{\log(1/\delta)}{\eps^{2}})$. For each level $t$, the accumulated budget $B_{t}$ is no more than $(72000 \cdot n \cdot \frac{\log(1/\delta)}{\eps}+K_{t} \cdot b)$, where $K_{t}$ is the number of arms appeared at level $t$. The overall budget is at most the summation of all the levels, which will be
\begin{align*}
\sum_{t=1}^{\frac{30}{\eps}} B_{t} &= \sum_{t=1}^{\frac{30}{\eps}} (72000 \cdot n \cdot \frac{\log(1/\delta)}{\eps}+K_{t} \cdot b) \\
&= 2160000 \cdot n \cdot \frac{\log(1/\delta)}{\eps^{2}} +  b\cdot \sum_{t=1}^{\frac{30}{\eps}} K_{t}\\
& \leq 2160000 \cdot n \cdot \frac{\log(1/\delta)}{\eps^{2}} +  b\cdot n \tag{the overall number of arms appeared cannot be greater than n}\\
&= 2160000 \cdot n \cdot \frac{\log(1/\delta)}{\eps^{2}} + 90000 \cdot (C+3) \cdot n \cdot \frac{\log(1/\delta)}{\eps^{2}} \tag{$\Delta=\frac{\eps}{30}$ in $\msubrout$}
\end{align*}
\noindent
And the above summation is in $O(n\cdot \frac{\log(1/\delta)}{\eps^{2}})$ since $C$ is a constant. Furthermore, once a $\king$ is defeated, we pay an additional $180\cdot \frac{1}{(\eps/30)^{2}}\cdot\log(1/\delta)$ arm pulls to estimate the empirical reward of the new $\king$. There are at most $\frac{30}{\eps}$ times for such a estimation to happen, so the sample complexity induced by this part is at most $180\cdot \frac{1}{(\eps/30)^{2}}\cdot\log(1/\delta)\cdot \frac{30}{\eps} = O\paren{\frac{1}{\eps^{3}}\cdot\log(1/\delta)}$.\par

On the other hand, the sample complexity of the \underline{Fraction Stream Procedure} is trivially $4\cdot \log(n)\cdot\frac{\log(1/\delta)}{\eps}\cdot 100 \cdot\frac{180}{\eps^{2}}\cdot\log(1/\delta)\cdot \log(n) = O\paren{\log^{2}(n)\cdot\frac{\log^{2}(1/\delta)}{\eps^{3}}}$. Taking together the above parts concludes the proof. 
\end{proof}
\noindent
We now show the correctness of the algorithm.
\begin{lemma}[Correctness]
\label{lem:2-arm-mab-correct}
With probability at least $(1-\delta)$, the algorithm returns an $\eps$-best arm.
\end{lemma}

We prove the correctness of the algorithm by looking into the cases when $|\mathcal{T}_{\pstar, \frac{\eps}{3}}|\leq\frac{\eps n}{\log(n)}$ and $|\mathcal{T}_{\pstar, \frac{\eps}{3}}|>\frac{\eps n}{\log(n)}$, respectively. On the high level, the correctness for the former case is guaranteed by the \underline{Full Stream Procedure} (\Cref{lem:alggeneral-full-correct}), while the correctness of the later case is guaranteed by the \underline{Fraction Stream Procedure} (\Cref{lem:alggeneral-frac-correct}). There are some `slacks' between the two cases, e.g. when $|\mathcal{T}_{\pstar, \frac{\eps}{3}}|=3\cdot\frac{\eps n}{\log(n)}$, it is possible for both procedures to return an $\eps$-best arm. It is easy to find the one with the higher reward in this type of scenario, though, and we do not pursue this direction here. 

We start with stating and proving the correctness of the $|\mathcal{T}_{\pstar, \frac{\eps}{3}}|\leq\frac{\eps n}{\log(n)}$ case.
\begin{lemma}[Full Stream Procedure]
\label{lem:alggeneral-full-correct}
Suppose $|\mathcal{T}_{\pstar, \frac{\eps}{3}}|\leq\frac{\eps n}{\log(n)}$, then the Full Stream Procedure returns an $\eps$-best arm with probability at least $(1-\delta)$. 
\end{lemma}
\noindent
We prove \Cref{lem:alggeneral-full-correct} by showing the following claims.
\begin{claim}
\label{clm:ptildeestimate}
Suppose $|\mathcal{T}_{\pstar, \frac{\eps}{3}}|\leq\frac{\eps n}{\log(n)}$, with probability at least $1-\frac{\delta}{2}$, the estimation $\ptilde$ eventually increases to the range of $(\pstar-\frac{1}{30}\eps)\leq\ptilde\leq(\pstar+\frac{1}{30}\eps)$, and the best arm $\arm^{*}$ eventually becomes the $\king$. 
\end{claim}
\noindent
\begin{proof}
We first show the lower bound of $\ptilde$. Consider the event when the best arm $\arm^{*}$ with reward $\pstar$ joins the stream, and suppose at this moment $\ptilde<(\pstar-\frac{1}{30}\eps)$, and there is $p_{1} = \ptilde - \frac{4}{30}\eps <\pstar-\frac{1}{10}\eps$. Then by \Cref{clm:challenge-soundness}, with probability at least $1-\frac{\delta}{16}$, the best arm always get a higher reward than the $p_{1}$. Therefore, it becomes the $\king$ by either defeat the previous $\king$ or exhaust the budget. Furthermore, since we are estimating the reward with $180\cdot \frac{1}{(\eps/30)^{2}}\cdot \log(\frac{1}{\delta})$ arm pulls, the probability for the empirical reward $\phat_{\arm^{*}}$ to be less than $\pstar-\frac{\eps}{30}$ can be bounded as
\begin{align*}
\Pr\paren{\pstar-\phat_{\arm^{*}} \geq \frac{\eps}{30}} & \leq \exp(-2\cdot 20\cdot \log(1/\delta)) \tag{by \Cref{prop:chernoff}}\\
&\leq \frac{\delta}{16}.
\end{align*}
Hence, with probability at least $(1-\frac{\delta}{8})$, eventually there is $\ptilde\geq (\pstar-\frac{1}{30}\eps)$, and the best arm becomes the $\king$ when $\ptilde<(\pstar-\frac{1}{30}\eps)$ \emph{before} the update.

We then show the upper bound of $\ptilde$. If $\ptilde$ is estimated by Line~\ref{line:increasetbytarget}, then with probability at least $(1-\frac{\delta}{16})$, $\ptilde$ will not be greater than $\pstar+\frac{1}{30}\eps$ since 
\begin{align*}
\Pr\paren{\ptilde\geq \pstar+\frac{1}{30}\eps \text{ from Line \ref{line:increasetbytarget}}} &\leq \Pr\paren{\phat_{\arm^{*}} -\pstar\geq \frac{\eps}{30}}\\
&\leq \exp(-2\cdot 20\cdot \log(1/\delta)) \tag{by \Cref{prop:chernoff}}\\
&\leq \frac{\delta}{16}.
\end{align*}
The remaining cases to handle is that we could have executed Line~\ref{line:increasetbyone} for too many times such that it leads to $\ptilde>(\pstar+\frac{1}{30}\eps)$. We show that this bad event is unlikely to happen thanks to \Cref{clm:arm-sample-limite} and \Cref{clm:challenge-arm-sample}. For $\ptilde$ to become more than $\pstar+\frac{1}{30}\cdot \eps$ by Line~\ref{line:increasetbyone}, there must be $\ptilde>\pstar$ before the update, therefore, $p_{1}>\pstar-\frac{4}{30}\eps$. By \Cref{clm:arm-sample-limite}, the number of arms pulls for the challenge is at most $540\cdot 100\cdot \log(1/\delta) \cdot \frac{\log(n)}{\eps^{2}}$ (setting $\Delta = \frac{\eps}{10}$). Since the Full Stream Procedure only deals with the case when $|\mathcal{T}_{\pstar, \frac{\eps}{3}}|\leq\frac{\eps n}{\log(n)}$, which also implies $|\mathcal{T}_{\pstar-\frac{4}{30}\eps, \frac{\eps}{10}}|\leq\frac{\eps n}{\log(n)}$, the budget for the top arms is sufficient for the challenge from $\mathcal{T}_{\pstar, \frac{\eps}{3}}$ arms. For the challenge from the $\mathcal{B}_{\pstar, \frac{\eps}{3}}$ arms (which are in $\mathcal{B}_{\pstar-\frac{4}{30}\eps, \frac{\eps}{10}}$), we can use \Cref{clm:challenge-arm-sample} to conclude that the budget of $K\cdot b$ (increasing for each arriving $\arm$) is sufficient with probability at least $1-\frac{\delta}{8}$, and Line~\ref{line:increasetbyone} will \emph{not} be executed by the challenge from the bottom arms. Therefore, the upper bound holds with probability at least $1-\frac{\delta}{4}$. \par

Finally, consider the best arm arrives \emph{after} $\ptilde\geq(\pstar-\frac{1}{30}\eps)$. Since there is $\ptilde \leq \pstar - \frac{1}{30}\eps$, we have $p_{1}\leq \pstar-\frac{\eps}{10}$. By \Cref{clm:challenge-soundness}, with probability at least $1-\frac{\delta}{8}$, the best arm always becomes the $\king$ by defeating the stored arm.

A union bound gives us that with probability at least $(1-\frac{\delta}{8}-\frac{\delta}{4}-\frac{\delta}{8})=(1-\frac{\delta}{2})$, the statement of \Cref{clm:ptildeestimate} holds.
\end{proof}
\begin{claim}
\label{clm:armtopsmallreturn}
Conditioning on the events of \Cref{clm:ptildeestimate}, the \underline{Full Stream Procedure} returns an $\frac{\eps}{5}$-best arm with probability at least $1-\frac{\delta}{2}$.
\end{claim}
\noindent
\begin{proof}
We consider the worst case when $\ptilde = (\pstar-\frac{1}{30}\eps)$, as other cases only makes the non-$\frac{\eps}{3}$-best arms more difficult to join the stream. Under this condition, there are three possible scenarios for the arriving arms. 
\begin{enumerate}
\item An arm with empirical reward of more than $\ptilde+\frac{1}{30}\eps$ joins the stream and updates $\ptilde$ by running Line~\ref{line:increasetbytarget}. Under this scenario, according to \Cref{clm:challenge-complete}, the probability for this arm not to be a $\frac{1}{30}\eps$-best arm is at most $\frac{\delta}{8}$.
\item An arm with a lower empirical reward challenges $\ptilde-\frac{4}{30}\eps$. According to \Cref{clm:challenge-complete}, the probability for any arm with empirical reward less than $(\ptilde-\frac{4}{30}\eps-\frac{1}{30}\eps)=(\pstar-\frac{1}{5}\eps)$ to replace the incumbent arm is at most $\frac{\delta}{4}$.
\item The budget of the $\king$ gets exhausted. According to \Cref{clm:challenge-arm-sample}, this does not happen with probability at least $1-\frac{\delta}{8}$ since we assume $|\mathcal{T}_{\pstar, \frac{\eps}{3}}|\leq \frac{\eps n}{\log(n)}$ (which implies $|\mathcal{T}_{\pstar-\frac{1}{5}\eps, \frac{\eps}{10}}|\leq \frac{\eps n}{\log(n)}$).
\end{enumerate}
Applying a union bound over the above cases gives us the success probability of at least $1-\frac{\delta}{2}$.
\end{proof}

\paragraph{Finalizing the proof of \Cref{lem:alggeneral-full-correct}.} Combining \Cref{clm:ptildeestimate} and \Cref{clm:armtopsmallreturn}, applying a union bound will get that the probability for a non-$\frac{\eps}{5}$-arm to be returned is at most $\delta$, which finalizes the proof.

We now prove the easier case for the Fraction Stream Procedure.
\begin{lemma}[Fraction Stream Procedure]
\label{lem:alggeneral-frac-correct}
Suppose $|\mathcal{T}_{\pstar, \frac{\eps}{3}}|>\frac{\eps n}{\log(n)}$, then the Fraction Stream Procedure will return an $\eps$-best arm with probability at least $1-\delta$. 
\end{lemma}
\begin{proof}
We first show that the arm-sampling line will give us at least one $\frac{\eps}{3}$-best arm with probability at least $1-\frac{\delta}{2}$. Since $|\mathcal{T}_{\pstar, \frac{\eps}{3}}|>\frac{\eps n}{\log(n)}$, the probability for any arm in $\mathcal{T}_{\pstar, \frac{\eps}{3}}$ to be selected is at least $\frac{\eps}{\log(n)}$. Now, the probability for \emph{no such arm} to be picked will be:
\begin{align*}
\Pr(\text{No arm in $\mathcal{T}_{\pstar, \frac{\eps}{3}}$ to be picked}) &\leq (1-\frac{\eps}{\log(n)})^{4\frac{\log(n)}{\eps}\log(\frac{1}{\delta})}\\
& \leq 3\cdot \exp(-(4+\log(\frac{1}{\delta}))) \tag{$\log(\frac{1}{\delta})\geq1$}\\
& \leq \frac{\delta}{2}
\end{align*}
\noindent
Consider comparing an $\frac{\eps}{3}$-best arm with a non-$\eps$-best arm. By pulling each arm $100\cdot \frac{180}{\eps^{2}}\cdot \log(n)$ time, we can prove that with probability at least $1-\frac{\delta}{4n}$, the best arm has a higher empirical reward than any non-$\eps$-best arm. Therefore, with probability at least $1-\frac{\delta}{4}$, the arm with the highest reward is an $\eps$-best arm. Applying a union bound over sampling and comparing success probabilities proves \Cref{lem:alggeneral-frac-correct}.
\end{proof}

\subsubsection*{Proof of \Cref{lem:2-arm-mab-correct}}
By \Cref{lem:alggeneral-full-correct} and \Cref{lem:alggeneral-frac-correct}, either the \underline{Full Stream Procedure} returns an $\eps$-best arm with probability $(1-\delta)$ when $|\mathcal{T}_{\pstar, \frac{\eps}{3}}|\leq\frac{\eps n}{\log(n)}$, or the \underline{Fraction Stream Procedure} returns an $\eps$-best arm with probability $(1-\delta)$ when $|\mathcal{T}_{\pstar, \frac{\eps}{3}}|>\frac{\eps n}{\log(n)}$. This proves \Cref{lem:2-arm-mab-correct}.

%% file: prelim-algs.tex

\newcommand{\prealga}{\ensuremath{\textnormal{\textsc{Filter-Select}}}\xspace}
\newcommand{\prealgb}{\ensuremath{\textnormal{\textsc{Filter-Top-Sampling}}}\xspace}
\newcommand{\prealgc}{\ensuremath{\textnormal{\textsc{Aggressive-Filter-Select}}}\xspace}

\newcommand{\coinh}{\ensuremath{\widetilde{\coin}}}

\section{Warm-Up: Simpler Algorithms for Finding Most Biased Coin}\label{sec:prelim-algs}

This section includes three simpler algorithms with asymptotically optimal sample complexity and $O(\log(n))$, $O(\log\log(n))$ and $\log^{*}(n)$ space complexity, respectively.  
These algorithms successively build on top of each other and involve addition of several new ideas that might be of independent interest. Moreover, they can be seen as a warm-up to our main algorithm in~\Cref{sec:main}.

\subsection{An $O(\log{n})$ Space Algorithm} 
We start by introducing the simplest algorithm with $O(\log(n))$ space complexity.  

\begin{proposition}\label{prop:prelim-alg1}
	There exists a streaming algorithm that given $n$ coins arriving in a stream with the gap parameter $\Delta$ and confidence parameter $\delta$, finds the most biased coin with probability at least $1-\delta$ 
	using $O(\frac{n}{\Delta^2} \cdot \log{(1/\delta)})$ coin tosses and a memory of $O(\log(n))$ coins. 
\end{proposition}

\paragraph{High Level Overview.} Our algorithm in this part is a streaming friendly implementation of the {median-elimination} algorithm of~\cite{EvenDarMM02} using the ``merge-and-reduce'' technique from
the streaming literature (see, e.g.~\cite{GuhaMMO00,AgarwalHV04}). We now give a high level overview of the algorithm. 

The idea behind the merge-and-reduce technique is as follows: suppose instead of storing all the coins in the memory and running median-elimination algorithm of~\cite{EvenDarMM02}, we 
store the first $\sqrt{n}$ coins; run median-elimination, pick the output coin of the algorithm, and discard the rest. We then read the next $\sqrt{n}$ coin in memory and do as before. This way, by the time we finished processing the stream, 
we have stored $\sqrt{n}$ additional coins (the output of median-elimination on each $\sqrt{n}$-size sub-stream). We run yet another median-elimination on these coins and return the output coin as the most biased coin. 
It is easy to verify that this algorithm can be implemented with $O(\sqrt{n})$ size memory and will output the correct answer with large (constant) probability.   

One can also recursively apply the idea above multiple times. For instance, we can pick the first $n^{1/3}$ coins in a \emph{bucket}, find their most biased coin and send it to the bucket at 
next \emph{level}, and once $n^{1/3}$ coins are collected in this bucket, do the same, and send the most biased coin among them to the the bucket of final level. This reduces the memory to $O(n^{1/3})$ coins now. In fact, by increasing the number of levels to 
$O(\log{n})$, we can reduce the number of coins we are storing in each bucket to some absolute constant and obtain an $O(\log{n})$ memory algorithm. There is a catch however with this approach: we need to do a union bound
over the $O(\log{n})$ times the (true) most biased coin participates in the median-elimination algorithm which increases the sample complexity of the algorithm by an $O(\log\log{n})$ factor and thus making it sub-optimal. 

There is however a simple fix to this: observe that the number of coins that participate in each level of this algorithm is dropping by a constant factor at each level. Hence, we can allocate more and more coin tosses to higher and higher levels in order to increase 
the probability of success on those levels, while still ensuring that the total sample complexity of the algorithm remains within the optimal range of $O(n)$ (after all, this is the same exact idea behind the median-elimination algorithm itself). 
This is precisely what our algorithm does. 

\begin{tbox}
	\textbf{The ${O(\log{n})}$ Space Algorithm:} 
	
	 \medskip
	
	Parameters ($s_\ell$ denotes the number of samples at level $\ell$):
	\begin{align*}
	\set{s_\ell}_{\ell \geq 1}: \quad s_{\ell} = \frac{4}{\Delta^2} \cdot \paren{\ln{(1/\delta)}+3^{\ell}}.
	\end{align*}
	
	Buckets: $B_{1}$, $B_{2}$, ..., $B_{t}$, each of size $4$ for $t :=\ceil{\log_4{(n)}}$. 
	\medskip
	
	\begin{itemize}
	\item For each arriving $\coin_i$ in the stream do:
	\begin{enumerate}[label=($\arabic*$)]
		\item Add $\coin_{i}$ to bucket $B_{1}$. 
		\item If any bucket $B_{\ell}$ is full: 
		\begin{enumerate}
			\item We sample each coin in $B_{i}$ for $s_{l}$ times;
			\item Select $\coin_{\ell}^{*}$ with the highest empirical bias and add it to $B_{\ell+1}$;
		\end{enumerate}
	\end{enumerate}
	\item At the end of the stream, select $\coin_{t}^{*}$ of bucket $B_{t}$ as the most biased coin.
	\end{itemize}
\end{tbox}

\begin{remark}\label{rem:fix-remark}
	Our algorithm is stated as if the number of coins is a power of $4$ or rather $\ceil{\log_4{(n)}} = \log_4{(n)}$. However, when this is not the case, the most biased coin may not have enough time to raise to the level $t$ itself. 
	There is a simple fix however: we can `pad' the stream with `dummy coins' which has $0$ bias until the stream length becomes a power of $4$. By doing so, the most biased coin, $\coinstar$ will have enough time to raise to the top level 
	and we simply prove in the following that this coin will not be dropped in any of the successive buckets with sufficiently large probability. The same idea can be used for our two other algorithms in this section as well (an alternative option would be to 
	run any standard algorithm, say median-elimination of~\cite{EvenDarMM02} on the set of $O(\log{n})$ coins stored across all buckets at the end of the stream; we omit the details). 
\end{remark}
\noindent
In practice, the algorithm can be implemented by checking if any bucket is full following a bottom-up manner. The following claim bounds the space complexity of this algorithm. 

\begin{claim}\label{clm:alg1-space}
	The space complexity of the algorithm is $O(\log(n))$. 
\end{claim}
\begin{proof}
	We maintain $t = O(\log{n})$ bucket each of size $O(1)$ throughout the stream. 
\end{proof}
We bound the sample complexity of the algorithm in the following lemma. 
\begin{lemma}\label{lem:alg1-sample}
The sample complexity of the algorithm is $O(\frac{n}{\Delta^2}\cdot\log{(1/\delta)})$.
\end{lemma}
\begin{proof}
By construction, the number of coins that ever appear in level $\ell$ (namely in bucket $B_\ell$) is bounded by $\frac{n}{4^{\ell-1}}$. The number of samples per each coin at level $\ell$ is also $3^{\ell}$. We thus have, 
\begin{align*}
	\textnormal{\# of samples} &= \sum_{\ell=1}^{t} \frac{n}{4^{\ell-1}} \cdot s_\ell =  \sum_{\ell=1}^{t} \frac{n}{4^{\ell-1}} \cdot \paren{\frac{4}{\Delta^2} \cdot \paren{\ln{(1/\delta)}} + 3^{\ell}} \\
	&= \paren{\frac{4n}{\Delta^2} \cdot \paren{\ln{(1/\delta)}}}  \cdot \sum_{\ell=1}^{t} \frac{1}{4^{\ell-1}} + \sum_{\ell=1}^{t} \frac{n}{4^{\ell-1}} \cdot 3^{\ell} \\
	&\leq  \paren{\frac{4n}{\Delta^2} \cdot \ln{(1/\delta)}} \cdot \frac{4}{3} + 12n \tag{as the first series converges to $4/3$ and the second to $12$ even when they go to infinity},
\end{align*}
which is $O(n\cdot\frac{\log(1/\delta)}{\Delta^{2}})$ as desired. 
\end{proof}

Finally, we prove the correctness of the algorithm.
\begin{lemma}\label{lem:alg1-correct}
With probability at least $1-\delta$, the algorithm returns the most biased coin.
\end{lemma}
\begin{proof}
Consider any bucket $B_t$ and assume that the most biased coin, $\coinstar$ is present in this bucket. The probability for any other coin, say $\coin_i$, to have a greater empirical bias than $\coinstar$ 
when we sample the coins in $B_t$ is at most, 
\begin{align*}
	\Pr\paren{\textnormal{$\coinstar$ has a lower empirical bias than $\coin_i$ in level $\ell$}} &\leq 2\exp\paren{-(\ln{(1/\delta)} + 3^{\ell})} \tag{by \Cref{lem:coin-comp} and choice of $s_\ell$ samples in this level} \\
	&\leq 2\delta \cdot \exp\paren{-3^{\ell}}.
\end{align*}

A union bound over the $4$ coins in bucket $B_t$ implies that the probability that $\coinstar$ is not returned at level $\ell$ is at most $8\delta \cdot \exp\paren{-3^{\ell}}$. By a union bound across all levels, 
we have, 
\begin{align*}
	\Pr\paren{\textnormal{$\coinstar$ is not returned as the answer}} &\leq \sum_{\ell=1}^{t} 8\delta \cdot \exp\paren{-3^{\ell}} < \delta \tag{as the series converges to $<0.05$ even when it goes to infinity}.
\end{align*}
This concludes the proof. 
\end{proof}

\subsection{An $O(\log\log{(n)})$ Space Algorithm} 

We now show how to tweak the $O(\log{n})$ space algorithm and reduce its space complexity exponentially, i.e., down to $O(\log\log{n})$. 

\begin{proposition}\label{prop:prelim-alg2}
	There exists a streaming algorithm that given $n$ coins arriving in a stream with the gap parameter $\Delta$ and confidence parameter $\delta$, finds the most biased coin with probability at least $1-\delta$ 
	using $O(\frac{n}{\Delta^2} \cdot \log{(1/\delta)})$ coin tosses and a memory of $O(\log\log(n))$ coins. 
\end{proposition}
\paragraph{High Level Overview.} Recall that the space complexity of the algorithm in~\Cref{prop:prelim-alg1} was governed by the number of the recursion levels (or elimination rounds) done by the algorithm which was $O(\log{(n)})$. 
As such, if we could somehow reduce the number of levels further, we should be able to reduce the space complexity as well (assuming we could still store only $O(1)$ coins per each level). We now explain how our algorithm achieves this. 

The idea is simple: Consider the level $\approx (\log\log{n})$ of the algorithm of~\Cref{prop:prelim-alg1}; by construction, only $O(n/\log{n})$ coins in the stream will ever make it to this level. This means that we can in fact spend $O(\log{n})$ samples
per these coins to have a very good estimate of their true bias using their empirical bias (since we can now do a union bound over \emph{all} these coins), and still remain within the $O(n)$ sample budget. Moreover, now that we are sampling each coin 
$O(\log{n})$ times, we can simply run the basic approach of just maintaining the current best coin (in terms of empirical bias) for the coins in this level -- this requires storing a single coin. As such, the space complexity of the algorithm 
is now $O(\log\log{n})$ (for storing the coins in the first $\approx (\log\log{n})$ levels) plus one extra coin (for storing the running max in the top level).

\begin{tbox}
	\textbf{The ${O(\log\log{(n)})}$ Space Algorithm:} 
	
	 \medskip
	
	Parameters ($s_\ell$ denotes the number of samples at level $\ell$, and $s_T$ is for the top most level):
	\begin{align*}
	\set{s_\ell}_{\ell \geq 1}: \quad s_{\ell} = \frac{4}{\Delta^2} \cdot \paren{\ln{(2/\delta)}+3^{\ell}}, \qquad s_T := \frac{4}{\Delta^2} \cdot \paren{\ln{(1/\delta)}+\ln{(n)}}.
	\end{align*}

	Buckets: $B_{1}$, $B_{2}$, ..., $B_{t}$ of size $4$ for $t:=\ceil{\log_{4}\ln{(n)}}$, and a single $\coinh$ as the candidate for the most biased coin. 
	
	\begin{itemize}
	\item For each arriving $\coin_i$ in the stream do:
	\begin{enumerate}[label=($\arabic*$)]
		\item Add $\coin_{i}$ to bucket $B_{1}$. 
		\item If any bucket $B_{\ell}$ is full: 
		\begin{enumerate}
			\item We sample each coin in $B_{i}$ for $s_{l}$ times;
			\item Select $\coin_{\ell}^{*}$ with the highest empirical bias and add it to $B_{\ell+1}$;
		\end{enumerate}
		\item For any coin $\coin^{*}_{t}$ as the most biased on the $t$-th level:
		\begin{enumerate}
			\item Sample the current candidate $\coinh$ and $\coin^{*}_{t}$ for $s_T$ times;
			\item Store the one with the higher empirical bias as the new $\coinh$;
		\end{enumerate}
	\end{enumerate}
	\item Return $\coinh$ after all the coins have been processed.

	\end{itemize}
\end{tbox}

See also~\Cref{rem:fix-remark} about the standard `padding argument' discussed earlier. 

\begin{claim}\label{clm:loglogn-space}
The space complexity of the algorithm is $O(\log\log(n))$.
\end{claim}
\begin{proof}
We maintain $\ceil{\log_{4}\ln(n)} = O(\log\log{n})$ buckets of size $4$ for the first $t$ levels and one extra coin space for the selection phase at the top. 
\end{proof}

\begin{lemma}\label{lem:alg2-sample}
The sample complexity of the algorithm is $O(\frac{n}{\Delta^2}\cdot\log{(1/\delta)})$.
\end{lemma}
\begin{proof}
The sample complexity incurred by the first part of the algorithm, namely, the $t$ levels of bucketing is already $O(\frac{n}{\Delta^2}\cdot\log{(1/\delta)})$ by \Cref{lem:alg1-sample} (by replacing $\delta$ with $\delta/2$). The only other part of sample complexity
is the one incurred in maintaining $\coinh$ in the top level. 

As the number of bucketing levels is $t$ and size of each bucket is $4$, only $n/4^{t}$ coins ever reach the top level. 
Any coin reaching to top level incur $2 \cdot s_T$ additional samples ($s_T$ for $\coinh$ and $s_T$ for the new coin), leading 
\begin{align*}
	\textnormal{\# of samples on top level} \leq 2 \cdot s_T \cdot n/4^t \leq \frac{8}{\Delta^2} \cdot \paren{\ln{(1/\delta)}+\ln{(n)}} \cdot n/\ln{n} \leq  \frac{8n}{\Delta^2} \cdot \paren{\ln{(1/\delta)}}, \tag{by the choice of $s_T$ and $t$}
\end{align*}
finalizing the proof. 
\end{proof}

\begin{lemma}\label{lem:alg2-correct}
With probability at least $1-\delta$, the returned $\coinh$  is the most biased coin $\coinstar$.
\end{lemma}

\begin{proof}

By \Cref{lem:alg1-correct} (by replacing $\delta$ with $\delta/2$), with probability at least $1-\delta/2$, $\coinstar$ will be preserved throughout the first $t$ levels of bucketing. As long as in any of the trials done in the top level, 
the empirical bias of $\coinstar$ is larger than any other coin, we are ensured that $\coinstar$ is returned as the correct answer.  Consider any other $\coin_i$ that reaches the top level. We have, 
\begin{align*}
	\Pr\paren{\textnormal{$\coinstar$ has a lower empirical bias than $\coin_i$ in top level}} &\leq 2\exp\paren{-(\ln{(1/\delta)} + \ln{(n)})} \tag{by \Cref{lem:coin-comp} and choice of $s_T$ samples in this level} \\
	&\leq \frac{2\delta}{n}.
\end{align*}

We can now do a union bound over at most $\frac{n}{\ln{n}} \leq \frac{n}{4}$ coins that reach the top level and obtain that the probability $\coinstar$ loses to any coin at this point is only $\delta/2$. A union bound 
on the two events above imply that with probability $1-\delta$ we return $\coinstar$ as the final answer. 
\end{proof}

\subsection{An $O(\logstar{(n)})$ Space Algorithm}\label{sec:logstar-coin}

This brings us to our final algorithm in this part with space complexity of $O(\logstar{(n)})$ coins. 

\begin{proposition}\label{prop:prelim-alg3}
	There exists a streaming algorithm that given $n$ coins arriving in a stream with the gap parameter $\Delta$ and confidence parameter $\delta$, finds the most biased coin with probability at least $1-\delta$ 
	using $O(\frac{n}{\Delta^2} \cdot \log{(1/\delta)})$ coin tosses and a memory of $\ceil{\log^{*}(n)}+1$ coins. 
\end{proposition}

\paragraph{High Level Overview.} Our algorithm in~\Cref{prop:prelim-alg2} suggested a way of discarding the entire $\log{n}-\log\log{n}$ levels of the original algorithm in~\Cref{prop:prelim-alg1} and replacing them by maintaining a 
simple running (candidate) best coin. 

To obtain the new algorithm, we recursively do this for every level of the algorithm of~\Cref{prop:prelim-alg1}, in effect, entirely bypassing the bucketing idea, and have a different \emph{leveling} scheme (for simplicity of exposition, we 
still refer to these at levels but note that these are different than levels of~\Cref{prop:prelim-alg1}). The important thing is that we no longer store an entire bucket per level to postpone the computation of their most biased coin to later. Instead, we compute a 
running (candidate) best coin in each level and once we visited ``enough'' number of coins in this level, we send this coin to the next level and do exactly the same. This way, we can consider a much larger number of coins per each level (by simply maintaining a \emph{counter}) without having to pay the cost of storing them explicitly.

\begin{tbox}
	\textbf{The $O(\logstar{(n)})$ Algorithm:} 
	
	\medskip
	
	Parameters ($s_\ell$ denotes the number of samples at level $\ell$, and $r_\ell$ specifies $s_\ell$):
	\begin{align*}
	&\set{r_{\ell}}_{\ell \geq 1}: \quad r_1 = 4, \quad r_{\ell+1} = 2^{r_{\ell}}; \tag{intermediate variables to define $s_\ell$ and $c_\ell$} \\
	&\set{s_\ell}_{\ell \geq 1}: \quad s_{\ell} = \frac{4}{\Delta^2} \cdot \paren{\ln{(1/\delta)}+ 3 \cdot r_\ell}; \tag{number of samples per each level} \\
	&\set{c_\ell}_{\ell \geq 1}: \quad c_\ell = \frac{2^{r_\ell}}{2^{\ell-1}} \tag{the bound for restarting the counter of each level}.
	\end{align*}
	
	Counters: $C_{1}, C_{2}, \ldots, C_{t}$ for $t= \ceil{\logstar{(n)}} + 1$. 
	
	\medskip
	
	Stored coins: $\coin^{*}_{1}, \coin^{*}_{2}, ..., \coin^{*}_{t}$ as the (candidate) most biased coin each level. 
	
	\begin{itemize}
	\item For each arriving $\coin_i$ in the stream do:
	\begin{enumerate}[label=($\arabic*$)]
		\item Starting from level $\ell=1$ to $t$ do: 
		\begin{enumerate}
			\item\label{line:sample}  Sample both $\coin_{i}$ and $\coin^{*}_{\ell}$ for $s_{\ell}$ times. If empirical bias of $\coin_i$ is less than $\coinstar_\ell$, drop $\coin_i$, otherwise, replace $\coinstar_\ell$ with $\coin_i$. 
			\item Increase $C_{\ell}$ by 1. If $C_{\ell}=c_{\ell}$, send $\coin^{*}_{\ell}$ to the next level by considering it as a new arriving coin in Line~\ref{line:sample} for $\ell+1$ and restart $C_\ell = 0$; otherwise 
			go to the next coin in the stream.  
		\end{enumerate}
	\end{enumerate}
	\item Return $\coin_{t}^{*}$ as the most biased coin. 
	\end{itemize}
\end{tbox}

See also~\Cref{rem:fix-remark} about the standard `padding argument' discussed earlier.

\begin{claim}\label{clm:alg3-space}
The space complexity of the algorithm is $\ceil{\log^{*}(n)} + 1$.
\end{claim}
\begin{proof}
We have $t = \ceil{\logstar{(n)}}+1$ levels, each containing a single coin. 
\end{proof}

\begin{lemma}\label{lem:alg3-sample}
The sample complexity of the algorithm is $O(\frac{n}{\Delta^2}\cdot\log{(1/\delta)})$.
\end{lemma}

\begin{proof}
Let $K_{\ell}$ denote the number of coins that are ever visited in level $\ell$. By construction, the sample complexity of the algorithm is:
\begin{align*}
	\textnormal{\# of samples} &= \sum_{\ell=1}^{t} K_\ell \cdot 2s_{\ell} \\
	&\leq \sum_{\ell=1}^{t} \frac{n}{\prod_{i=1}^{\ell-1}c_{i}} \cdot 2s_{\ell} \tag{as for each $c_{\ell'}$ coin in a level $\ell'$, we only send one coin to the level $\ell'+1$} \\
	&\leq \sum_{\ell=1}^{t} \frac{n}{c_{\ell-1} \cdot c_{\ell-2}} \cdot 2s_{\ell} \tag{where we define $c_{-1} = c_{0} = 1$} \\
	&=  \sum_{\ell=1}^{t} \frac{n}{c_{\ell-1} \cdot c_{\ell-2}} \cdot \frac{8}{\Delta^2} \cdot \paren{\ln{(1/\delta)}+ 3 \cdot r_\ell} \tag{by the choice of $s_\ell$} \\
	&\leq \Paren{\frac{8n}{\Delta^2}} \cdot \Paren{\ln{(1/\delta)} \cdot \sum_{\ell=1}^{t} \frac{1}{c_{\ell-1} \cdot c_{\ell-2}} + \sum_{\ell=1}^{t} \frac{3r_\ell}{c_{\ell-1} \cdot c_{\ell-2}}} \\
	&\leq \Paren{\frac{8n}{\Delta^2}} \cdot \Paren{2\ln{(1/\delta)} + \sum_{\ell=1}^{t} \frac{3r_\ell}{c_{\ell-1} \cdot c_{\ell-2}}}  \tag{the first series converges to $<2$ in infinity}  \\
	&\leq \Paren{\frac{8n}{\Delta^2}} \cdot \Paren{2\ln{(1/\delta)} + O(1) + \sum_{\ell=3}^{t} \frac{3 \cdot 2^{\ell-2}}{c_{\ell-2}}}  \tag{by the choice of $c_{\ell-1}$ and $r_{\ell} = 2^{r_{\ell-1}}$, and since $r_1,r_2=O(1)$} \\
	&= O(\frac{n}{\Delta^2}\cdot\log{(1/\delta)}) \tag{the second series converges to $O(1)$ also in infinity}. 
\end{align*}

It is also worth mentioning here that by the choice of $t = \ceil{\logstar{(n)}} + 1$, $K_{t+1} = 0$ and hence the algorithm never finishes processing its last level (which is required for its correctness). 
\end{proof}

\begin{lemma}\label{lem:alg3-correct}
With probability at least $1-\delta$, the algorithm returns the most biased coin.
\end{lemma}

\begin{proof}
Consider any level $\ell$ and assume the most biased coin $\coinstar$ is present in this level. Then, the probability that any other coin $\coin_i$ has a greater empirical bias than $\coinstar$ is at most, 
\begin{align*}
	\Pr\paren{\textnormal{$\coinstar$ has a lower empirical bias than $\coin_i$ in level $\ell$}} &\leq 2\exp\paren{-(\ln{(1/\delta)} + 3r_\ell)} \tag{by \Cref{lem:coin-comp} and choice of $s_\ell$ samples in this level} \\
	&\leq 2\delta \exp\paren{-3r_\ell} \leq \frac{\delta}{4 \cdot 2^{r_\ell}}. 
\end{align*}
On the other hand, the total number of coins that will be compared with $\coinstar$ at level $\ell$ (before the counter gets reset) is at most $c_\ell = \frac{2^{r_{\ell}}}{2^{\ell-1}}$. Hence, by a union bound, 
the probability that $\coinstar$ loses to any of them is at most $\frac{\delta}{2^{\ell+1}}$. This means that assuming $\coinstar$ is present at level $\ell$, the probability that it is not sent to the next level is only $\frac{{\delta}}{2^{\ell+1}}$. 
We can now do a union bound over all levels and obtain that: 
\begin{align*}
	\Pr\paren{\textnormal{$\coinstar$ is not returned as the answer}} &\leq \sum_{\ell=1}^{t} \frac{{\delta}}{2^{\ell+1}} \leq \delta \tag{as this series converges to $1$ in infinity}. 
\end{align*}
This concludes the proof. 
\end{proof}

%% file: random-walk.tex
\newcommand{\calg}{\ensuremath{\textnormal{\textsf{ALG-c}}}\xspace}

\section{Random Walk with Flexible Step Size}\label{sec:rand-walk}

In the proof of \Cref{lem:main-keep}, we have shown that with the `conservative' challenging rules, the number of coin tosses never exhausts the cumulative budget over the $\Theta(n)$ stream. Notice that the challenge process can be viewed as the fluctuation of a random variable with deterministic increment steps (`increase budget') and randomized decreasing steps (`coin tosses'). In this sense, the challenging process can be perceived as a variation of a classical \emph{Random Walk}, which concerns the value of a random sequence with certain probabilities for walking `forward' and `backward'. In this section, we will look into more details about the random walk and study the characteristics of the coin challenge process in $\mainalg$ from this perspective.
\subsection*{Classical Random Walk}
We first give the definition of a classical one-dimension random walk. 

\begin{definition}[One-dimensional Random Walk]
\label{def:classic-walk}
A one-dimension random walk with $n$ steps and forward-moving probability $p$ is a stochastic sequence $\set{S_{i}}_{i=0}^{n}$ with the following characteristics: In the beginning, $S_{0}=0$; At each step $i\in [n]$, $S_{i}=\sum_{j=1}^{i}X_{j}$ with 
independent random variables $X_{j} \in \set{-1,+1}$ such that $\Pr(X_{j}=1)=p$. 
\end{definition}
\noindent
The following proposition is well-known. 

\begin{proposition}[Non-negativity of One-Dimensional Random Walk]
\label{prop:classic-walk}
For a one-dimension random walk characterized by \Cref{def:classic-walk} and $p>\frac{1}{2}$, we have:
\begin{equation*}
\Pr\paren{\exists i: S_{i}\leq 0}< O(1-p). 
\end{equation*}
\end{proposition}

\subsection*{Flexible Step-length Random Walk}

Notice that for \Cref{prop:classic-walk} to hold, the steps in a classical random walk can only be $+1$ or $-1$. A natural question to ask is that if the step length becomes flexible and unbounded, how can we keep the quantity of a positive? In this section, we will discuss the properties of such type of walk and its relationship with the challenging process of $\mainalg$ and \Cref{lem:main-keep}. 


In general, to keep the quantity of a walk positive with unbounded step lengths, there are the following two properties to consider:
\begin{enumerate}
\item The expectation of the backward step size should be smaller than the forward step size. Therefore, in expectation, the walk will be positive in quantity.
\item The variance is not too large. Hence, even if the backward steps become larger, it will not `exhaust' all the accumulated forward steps very quickly.
\end{enumerate}
Based on these, we define a \emph{Flex-length Positive Random Walk} as the following process:
\begin{definition}[Flex-length Positive Random Walk]
\label{def:challenge-walk}
A \emph{Flex-length Positive Random Walk}  with $n$ steps is a stochastic sequence $\set{S_{i}}_{i=0}^{n}$ with the following characteristics: In the beginning, $S_{0}=0$; At each step $i\in [n]$, $S_{i}=\sum_{j=1}^{i}X_{j}$ with independent random variables 
$X_{j}$ with the following properties:
\begin{enumerate}
\item $X_{j}$ is a sub-exponential random variable with parameter $\kappa=\frac{1}{\ln(1/\delta)}$.
\item $\expect{X_j} \geq \eta(j)$, where $\eta(j) :=C\cdot \frac{\ln(j/\delta)}{\sqrt{j}}$ for some absolute constant $C > 0$. 
\end{enumerate}
\end{definition}

We now prove an analogue of~\Cref{prop:classic-walk} for the Flex-length Positive Random Walk. 

\begin{proposition}[Non-negativity of Flex-length Positive Random Walk]
\label{prop:flexible-walk}
With probability at least $(1-\delta)$, \emph{Flex-length Positive Random Walk} will have $S_{i}>0$ for all  $i\in [n]$.
\end{proposition}

\begin{proof}
We prove the lemma by showing that by the choice of the parameter $\kappa$, the quantity of the walk will never derive more than $O(\sqrt{i}\log(i))$ away from its expectation. Formally, we will show:
\begin{align*}
\Pr\paren{\exists i: \card{S_{i}-\expect{S_{i}}}\geq C\cdot\sqrt{i}\cdot\ln(\frac{i}{\delta})}\leq\delta.
\end{align*}
Notice that $S_{i} = \sum_{j=1}^{i}X_{j}$. Thus, by the linearity of expectation, we have:
\begin{align*}
{S_{i}-\expect{S_{i}}} & = \paren{\sum_{j=1}^{i}X_{j} - \expect{\sum_{j=1}^{i}X_{j}}}\\
&=\sum_{j=1}^{i}\paren{X_{j}-\expect{X_{j}}}.
\end{align*}
Denote $X'_{j}=\paren{X_{j}-\expect{X_{j}}}$, and apparently $X'_{j}$ will be zero-mean. Recall that $X_{j}$ are sub-exponential random variables (and so are $X'_j$'s); Thus, by Bernstein's inequality (\Cref{prop:bernstein}):
\begin{align*}
	\Pr\paren{\card{S_i - \expect{S_i}} \geq C \cdot \sqrt{i}\cdot\ln(\frac{i}{\delta})} &= \Pr\paren{\card{\sum_{j=1}^{i}X'_{j}}\geq C \cdot \sqrt{i}\cdot\ln(\frac{i}{\delta})}\\ 
	&\leq 2 \cdot \exp\paren{-c \cdot \min\paren{\frac{C^{2} \cdot\ln^{2}(i/\delta)}{\kappa^2}, \frac{C \cdot \sqrt{i}\cdot\ln(i/\delta)}{\kappa}}} \tag{$c>0$ is a constant} \\
\end{align*}

The smaller one in the $\min(\cdot,\cdot)$ term will be dependent on $i$. Define $\hat{i}:= \paren{C\cdot\frac{\ln(\hat{i}/\delta)}{\kappa}}^2$. As such, the second term above will be the minimum whenever $i \leq \hat{i}$. 
Specifically, we have:

\begin{align*}
\Pr\paren{\card{S_i - \expect{S_i}} \geq C \cdot \sqrt{i}\cdot\ln(\frac{i}{\delta})\;\middle|\;i<\hat{i}} &\leq 2\cdot \exp\paren{-c\cdot\frac{C \cdot \sqrt{i}\cdot\ln(i/\delta)}{\kappa}} \tag{since $i$ is small}\\
&\leq \frac{\delta}{4}\cdot\frac{\exp(-\sqrt{i})}{i^{2}}. \tag{by the value of $\kappa$ and picking a sufficiently large $C$}
\end{align*}

Also, for the $i\geq \hat{i}$, we will have:
\begin{align*}
\Pr\paren{\card{S_i - \expect{S_i}} \geq C \cdot \sqrt{i}\cdot\ln(\frac{i}{\delta})\;\middle|\;i\geq\hat{i}} &\leq 2\cdot \exp\paren{-c\cdot\frac{C^{2} \cdot\ln^{2}(i/\delta)}{\kappa^2}} \tag{since $i$ is large}\\
&\leq \frac{\delta}{4}\cdot\frac{1}{i^{2}}. \tag{by the value of $\kappa$ and picking a sufficiently large $C$}
\end{align*}

By a union bound for all the choices of $i$:
\begin{align*}
	\Pr\paren{\exists i: \card{S_i - \expect{S_i}} \geq C \cdot \sqrt{i}\cdot\ln(\frac{i}{\delta})} &\leq (\delta/4) \cdot\paren{ \sum_{i=1}^{\hat{i}-1} \frac{\exp\paren{-\sqrt{i}}}{i^{2}} +\sum_{i=\hat{i}}^{n}\frac{1}{i^{2}}}\\
	&\leq (\delta/4) \cdot\sum_{i=1}^{n}\frac{1}{i^{2}} \tag{as $\exp(-\sqrt{i})\leq 1$}\\
	&< (\delta/2). \tag{as this series converges to $<2$}
\end{align*}
which proves the distance between $S_{i}$ and its expectation $\expect{S_{i}}$ for \emph{any} $i$ can only be at most $C\cdot\sqrt{i}\cdot\ln(\frac{i}{\delta})$ with $1-\delta$ probability. \par\noindent
Finally, since we have the expectation of each $X_{j}$ is at least $C\cdot\frac{\ln(j/\delta)}{\sqrt{j}}$, we should have 
\[
\expect{S_{i}}>C\cdot\sum_{j=1}^{i}\frac{\ln(j/\delta)}{\sqrt{j}}\geq C\cdot\sqrt{i} \cdot {\ln(i/\delta)},
\]
 since $\frac{\ln(j/\delta)}{\sqrt{j}}$ decreases monotonously for $j\geq 1$. Therefore, we have $\expect{S_{i}}>C\cdot\sqrt{i}\cdot\ln(\frac{i}{\delta})$ and the proof can be finalized.
\end{proof}

\begin{remark}
We remark that for the quantity of the walk never goes back to $0$, a weaker condition of $\sum_{j=1}^{i}\eta(j)>C\cdot\sqrt{i}\cdot\ln(\frac{i}{\delta})$ is sufficient. Also, for general $\eta$ without any restriction, we can show by parameter substitution (changing $\delta$ to the function of $\eta$) that the probability for the quantity of the walk decreasing to $0$ is at most $2\cdot\exp(-\frac{\eta}{C})$, which decreases exponentially as $\eta$ becomes larger.
\end{remark}

Based on \Cref{prop:flexible-walk}, we can re-formulate \Cref{lem:main-keep} as a special type of \emph{Flex-length Positive Random Walk} with stronger parameter conditions. This also gives a more systematic explanation on why $\mainalg$ holds.

\begin{proposition}[Reformulation of \Cref{lem:main-keep}]
\label{prop:coin-walk-vanilla}
The \emph{Challenge subroutine} in $\mainalg$ forms a \emph{Flex-length Positive Random Walk} with $\eta(j)>C\cdot\ln(\frac{1}{\delta})$ for all $j > 0$ and $\kappa=\frac{15}{\ln(1/\delta)}$. Thus, the quantity of the walk never decrease to $0$ with probability at least $1-\delta$.
\end{proposition}
\begin{proof}
Recall that at each step $i$, we will \emph{surely} accumulate $C\cdot\frac{4}{\Delta^{2}}\cdot\ln(\frac{1}{\delta}) + s_{1}$ budgets and use $s_{1}$ amount of them. Also recall that in \Cref{clm:main-expect}, we showed the expected number of coin tosses other than $s_{1}$ is less than $1$. Thus, the expectation of $X_{j}$ on any step $j$ should be more than $C\cdot\frac{4}{\Delta^{2}}\cdot\ln(\frac{1}{\delta})-1>C\cdot\ln(\frac{1}{\delta})$. Now observe the $\eta$ parameter is greater than $C\cdot \frac{\ln(i/\delta)}{\sqrt{i}}$ already, and the $\kappa$ parameter is also stronger than the requirement, so the quantity of the walk will never decrease to $0$ with probability at least $1-\delta$.
\end{proof}

We can actually draw a comparison between the `walk' in \Cref{lem:main-keep} and a classical random walk. The difference can be illustrated as figure \ref{fig:walkcomp}.
\begin{figure}[h!]
\centering
\includegraphics[width=1.0\textwidth]{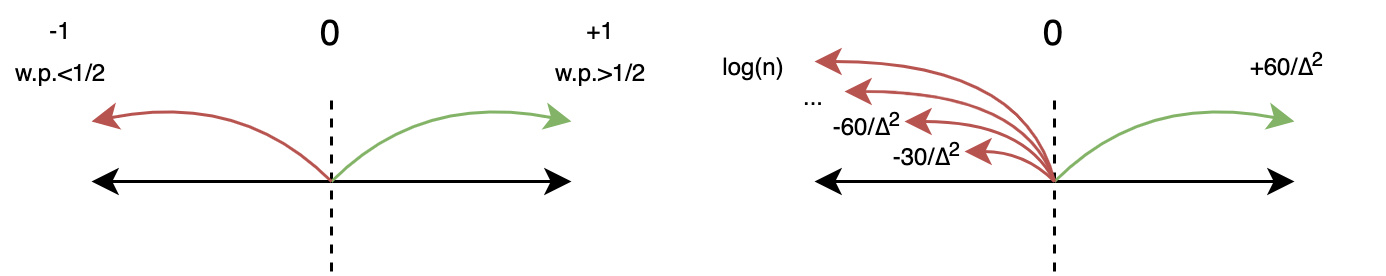}
\caption{\label{fig:walkcomp}Comparison between a classical random walk with $p>1/2$ (left) and a Coin-Game walk with the \emph{Challenge subroutine} as the challenging algorithm (right). }
\end{figure}

From the figure, it can be found that where are two major differences between a classical random walk and the walk in \Cref{lem:main-keep}. The first difference is that at each step, the challenge process will both increase and decrease the quantity of the walk deterministically; The second difference is that the step size of the backward walks in the challenge process is a function of the challenge rules and is randomized. A crucial observation to guarantee the correctness is the challenge subroutine in $\mainalg$ provides a sub-exponential distribution for the randomized backward step.